\documentclass[reqno,11pt]{amsart}
\usepackage{amsmath, latexsym, amsfonts, amssymb, amsthm, amscd}
\usepackage[rightcaption]{sidecap}
\usepackage{multirow,stmaryrd,appendix}
\usepackage{color}
\usepackage{graphicx,epsf,psfrag,caption,subcaption}
\usepackage[hidelinks, pdfstartview=FitV, pagebackref=false]{hyperref}

\usepackage{enumerate}
\usepackage[shortlabels]{enumitem}

\numberwithin{equation}{section}

\newtheorem{theorem}{Theorem}[section]
\newtheorem{lemma}[theorem]{Lemma}
\newtheorem{proposition}[theorem]{Proposition}

\newtheorem{rem}[theorem]{Remark}

\newtheorem{definition}[theorem]{Definition}

\newcommand{\ind}{\mathbf{1}}

\renewcommand{\tilde}{\widetilde}

\newcommand{\cF}{{\ensuremath{\mathcal F}} }
\newcommand{\cP}{{\ensuremath{\mathcal P}} }

\newcommand{\cN}{{\ensuremath{\mathcal N}} }

\newcommand{\cZ}{{\ensuremath{\mathcal Z}} }

\newcommand{\bP}{{\ensuremath{\mathbf P}} }
\newcommand{\bQ}{{\ensuremath{\mathbf Q}} }
\newcommand{\bE}{{\ensuremath{\mathbf E}} }

\DeclareMathSymbol{\leqslant}{\mathalpha}{AMSa}{"36} 
\DeclareMathSymbol{\geqslant}{\mathalpha}{AMSa}{"3E} 
\DeclareMathSymbol{\eset}{\mathalpha}{AMSb}{"3F}     
\newcommand{\dd}{\,\text{\rm d}}             
\newcommand{\suptwo}[2]{\sup_{\substack{#1 \\ #2}}} 
\newcommand{\sumtwo}[2]{\sum_{\substack{#1 \\ #2}}} 
\newcommand{\limtwo}[2]{\lim_{\substack{#1 \\ #2}}}     
\newcommand{\liminftwo}[2]{\liminf_{\substack{#1 \\ #2}}} 

\newcommand{\bbE}{{\ensuremath{\mathbb E}} }

\newcommand{\bbL}{{\ensuremath{\mathbb L}} }

\newcommand{\bbN}{{\ensuremath{\mathbb N}} }

\newcommand{\bbP}{{\ensuremath{\mathbb P}} }

\newcommand{\bbR}{{\ensuremath{\mathbb R}} }

\newcommand{\ga}{\alpha}
\newcommand{\gb}{\beta}
\newcommand{\gd}{\delta}
\newcommand{\gep}{\varepsilon}       

\newcommand{\gG}{\Gamma}

\newcommand{\gD}{\Delta}

\newcommand{\go}{\omega}

\newcommand{\gl}{\lambda}

\newcommand{\gs}{\sigma}

\makeatletter
\def\captionfont@{\footnotesize}
\def\captionheadfont@{\scshape}

\long\def\@makecaption#1#2{%
  \vspace{2mm}
  \setbox\@tempboxa\vbox{\color@setgroup
    \advance\hsize-6pc\noindent
    \captionfont@\captionheadfont@#1\@xp\@ifnotempty\@xp
        {\@cdr#2\@nil}{.\captionfont@\upshape\enspace#2}%
    \unskip\kern-6pc\par
    \global\setbox\@ne\lastbox\color@endgroup}%
  \ifhbox\@ne 
    \setbox\@ne\hbox{\unhbox\@ne\unskip\unskip\unpenalty\unkern}%
  \fi
  \ifdim\wd\@tempboxa=\z@ 
    \setbox\@ne\hbox to\columnwidth{\hss\kern-6pc\box\@ne\hss}%
  \else 
    \setbox\@ne\vbox{\unvbox\@tempboxa\parskip\z@skip
        \noindent\unhbox\@ne\advance\hsize-6pc\par}%
\fi
  \ifnum\@tempcnta<64 
    \addvspace\abovecaptionskip
    \moveright 3pc\box\@ne
  \else 
    \moveright 3pc\box\@ne
    \nobreak
    \vskip\belowcaptionskip
  \fi
\relax
}
\makeatother
\def\writefig#1 #2 #3 {\rlap{\kern #1 truecm
\raise #2 truecm \hbox{#3}}}

\newcommand{\tf}{\textsc{f}}
\newcommand{\tg}{\textsc{g}}
\newcommand{\tl}{\textsc{l}}

\newcommand{\var}{\mathrm{Var}}

\begin{document}

\title[Generalized pinning model]{
Localization,  big-jump regime  and  the effect disorder for
a class of generalized pinning models}

\author{Giambattista Giacomin}
\address{
  Universit\'e de Paris,   Laboratoire de Probabilit{\'e}s, Statistique et Mod\'elisation, UMR 8001,
            F- 75205 Paris, France
}

\author{Benjamin Havret}
\address{
  Universit\'e de Paris,   Laboratoire de Probabilit{\'e}s, Statistique et Mod\'elisation, UMR 8001,
            F- 75205 Paris, France
}

\begin{abstract}
One dimensional pinning models 
have been widely studied in the physical and mathematical literature, also in presence of disorder.
Roughly speaking, they undergo a transition between a delocalized phase 
and a localized one. In mathematical terms these models are  obtained by modifying the distribution of a discrete renewal process via a Boltzmann factor with an energy that contains only one body potentials. 
For some more complex models, notably pinning models based on higher dimensional renewals, it has been shown that 
other phases may be present. 

We study a generalization of the one dimensional pinning model in which the energy may depend in a nonlinear way on the contact fraction: this class of models contains the circular DNA case considered for example in  \cite{cf:bar2}.
We give a full solution of this generalized pinning model in absence of disorder and show that another transition appears. In fact the systems  may display
up to three different regimes: delocalization, partial localization and full localization.  What happens in the  partially localized regime 
can be explained in terms of the
``big-jump" phenomenon for  sums 
of heavy tail random variables under conditioning.

We then show that disorder  completely smears this second transition and we are back to the delocalization versus localization scenario. In fact we show that 
 the disorder, even if arbitrarily weak, is incompatible with the presence of a big-jump. 

\end{abstract}


\maketitle


\section{Introduction of the model and results}

\subsection{Phase transitions, disorder and   pinning models}
\label{sec:preintro}

The pinning model, sometimes called Poland-Scheraga  model,  comes up  in a variety of  real world phenomena. For example in the context of  DNA denaturation (this is the Poland-Scheraga  framework \cite{cf:PSbook}), for polymers
in presence of a \emph{defect region} \cite{cf:Fisher,cf:GB,cf:dH},
 for one dimensional interfaces in two dimensional systems with suitable boundary conditions \cite{cf:Velenik}.  But pinning models  have also an intrinsic and theoretical interest, due in particular to the following
crucial features:

\medskip

\begin{itemize}
\item the model is solvable in its homogeneous version: with this respect we cite in particular \cite{cf:Fisher}, but, as pointed out in \cite[App.~A]{cf:GT06}, the solvability mechanism is in reality just the basics  of Renewal Theory developed in mathematics  since the 40s with seminal contributions by 
J.~L.~Doob, P.~Erd\"os, W.~Feller and many others (e.g. \cite[App.~A]{cf:GB} and references therein). Unless we specify otherwise, when we speak of pinning models, like here, we mean \emph{one dimensional} pinning models: there are several higher dimensional generalizations that can and have been considered (e.g. \cite{cf:GLsharp} and references therein), and a class is going to be very relevant to us  and will  soon be mentioned;
\item the model exhibits a transition between a delocalized and a localized regime which is  understood in depth thanks to  solvability. Notably,
the model depends on a real parameter $\ga \ge 0$ and the critical phenomenon depends on $\ga$ is such a way that 
the (de)localization transition can be of arbitrary order, i.e. from first order (no differentiability) to infinite order ($C^\infty$ transition).
In the physical literature   the parameter $c=1+\ga$ is typically used, but we are going to stick to  $\ga$ for the natural link with 
the stable law exponent of the inter-arrival distribution for the underlying renewal process.
\end{itemize}

\medskip

Connected to the features we just outlined two research directions are particularly relevant to us:
\medskip

\begin{enumerate}
\item the effect of disorder on pinning models has been widely studied (\cite{cf:DHV,cf:G,cf:IM} and references therein), both because of its importance from the modeling standpoint and because of the whole spectrum of critical phenomena  generated  by tuning the parameter $\ga$.  In fact,
 understanding 
 the effect of disorder on criticality is an important general issue: the stability of criticality under the introduction of disorder is expected to depend on the critical behavior in the homogeneous system. In a nutshell, less singular transitions are expected/predicted to be more stable. 
 This is notably the content of the so called 
 Harris criterion for  disorder irrelevance \cite{cf:Harris} to which we  come back below (see in particular Remark~\ref{rem:Harris}).
\item (one dimensional) homogeneous pinning models exhibit only one transition, the (de)localization one. But it has been shown that some \emph{generalized} pinning models may exhibit other transitions: this is in particular the case of the generalized Poland-Scheraga  (gPS) model that takes into account the fact that the two DNA strands may have different length and that the pairing between the two strands may not be perfect \cite{cf:GOO,cf:GO04,cf:EON11}. As pointed out in \cite{cf:GK}, the gPS model can be seen as  a pinning model based on a two dimensional renewal and its solvability nature (first pointed out in \cite{cf:GOO}) can once again be seen in renewal theory terms. The novel transition exhibited by the gPS is interpreted in the physical literature in analogy with \emph{condensation phenomena} \cite{cf:EON11}. In renewal theory terms it is a phenomenon    for conditioned sums of independent heavy tail random variables that goes under the name of \emph{big-jump} \cite{cf:AL,cf:BGK,cf:DDS08}. The big-jump phenomenon has attracted attention in the mathematical community also in connection with other condensation phenomena
(see \cite{cf:FLS} and references therein). We refer to \cite{cf:VBB}, and references therein, for more on big-jump regimes in physical systems.  
\end{enumerate}

\medskip

We consider a  generalization of the pinning model which is simpler than the ones just mentioned in (2). This model is  based on one dimensional renewals but, in the homogeneous set-up, it  exhibits 
a condensation/big-jump transition, in addition to  the (de)localization transition. The circular DNA models studied in \cite{cf:RuBr,cf:bar1,cf:bar2} have been one of the motivations of our work and appear as a particular case of the family we study.  

We study the effect of the disorder on this class of models and our main result 
is that condensation/big-jump transitions do not withstand the introduction of disorder. By this we mean  that  the transition  is completely washed out and  there is no condensation/big-jump  in presence of disorder, even an arbitrarily weak disorder.
In the Harris criterion language, 
disorder is therefore \emph{relevant}, and in a very drastic way, even if, as we will explain, big-jump transitions are rather smooth transitions and a direct application of Harris criterion \cite{cf:Harris} does not suggest disorder relevance. 

\subsection{The generalized pinning model}
The model is based on the discrete renewal process $\tau=(\tau_j)_{j=0,1, \ldots}$ with $\tau_0=0$,
that is, if we set $\eta_j:=\tau_{j}-\tau_{j-1}$ we have that  $(\eta _j)_{j=1,2, \ldots}$ is an IID sequence taking values in $\bbN:=\{1,2, \ldots\}$. By using $\eta$ for  $\eta_j$ we set $K(n):=\bP(\eta=n)$ and assume that
\begin{equation}
\label{eq:K}
K(n) \stackrel{n\to \infty} \sim \frac{C_K}{n^{1+\ga}}\, ,
\end{equation}
for $\ga>0$ and a positive constant $C_K$.
While mathematically not really necessary, we assume that $K(n)>0$ for every $n\in \bbN$: this  does simplify some proofs and it 
is assumed in part of the literature that we use. We point out that for 
the (bio-)physical interpretation of the model   $K(1)>0$ is very natural (e.g. \cite[Ch.~1]{cf:GB})
We stress that $\sum_{n \in \bbN} K(n)=1$ and note that $\bE[\eta] \in (1, \infty)$ if $\ga >1$, while 
$\bE[\eta]=\infty$ if $\ga\in (0,1]$. The generalization to regularly varying $K(\cdot)$, see e.g. \cite[App. A.4]{cf:GB} is possible \cite[Ch.~4]{cf:thesis-BH} and  in most cases it is straightforward. But it carries a certain burden of notations and technicalities that cannot be motivated in terms of new phenomena.

\medskip

We are going to see $\tau$ as a random subset of $\bbN \cup\{0\}$, which almost surely contains $0$ and infinitely many other points. In particular $\gd_n:= \ind_{n \in \tau}$ is the indicator function that there exists $j$ such that $\tau_j=n$ and 
$\vert \tau \cap (0, N]\vert= \sum_{n=1}^N \gd_n= \sup\{j=1,2, \ldots:\, \tau_j \le N \}$. 

The class of models we present is based on a function $(m, N) \mapsto \Psi (m, N)$ defined for $N \in \bbN$ and $m \in \{ 1, \ldots, N\}$. We give here the  conditions we require on $\Psi$:

\medskip

\begin{definition}
\label{def:Psi}
We assume that 
\begin{equation}
\label{eq:defPsi0}
\Psi(m, N)\, =\, Q(m, N)\exp\left( N H(m/N)\right)\, , 
\end{equation}
with
\begin{enumerate}
\item $H: [0,1] \to \bbR \cup \{-\infty\}$  concave and real analytic in  the interior  of its domain; moreover 
we assume that $H$ is continuous up to the boundary points, including (with abuse of notation) the possibility that 
$H(0)$ and/or $H(1)$ are equal to $-\infty$.
\item $Q(m, N)\ge 0$ and 
\smallskip
\begin{itemize}
\item for every $b>0$ there exists $c>0$ such that for every
$N$ and $m\le N$ we have 
\begin{equation}
\label{eq:defPsi3}
Q(m, N)\, \le \, c \exp\left(b N\right)\, ;
\end{equation}
\item  for every $u, v \in (0,1)$, $u< v$, and every $b>0$ there exists $c>0$ such that for every
$N$ and $m$ with $m/N \in [u,v]$ we have
\begin{equation}
\label{eq:defPsi1}
Q(m, N)\, \ge  \, c \exp\left(-b N\right)\, .
\end{equation}
\end{itemize}
\end{enumerate}
\end{definition}

\medskip

These conditions  readily imply that for $\rho \in (0,1)$
\begin{equation}
\label{eq:fromdefPsi}
\limtwo{N \to \infty:}{m/N \to \rho} \frac 1N \log \Psi(m, N)\, =\, H(\rho)\, .
\end{equation}
Moreover, we will say that $H$ is trivial if $H''(\rho)=0$ for every $\rho\in (0,1)$: this is the case of
$H$ affine. Whenever $H$ is not trivial, $H$ is strictly convex because it is analytic. 

\medskip

While a full analysis is possible, to keep reasonably concise in  the analysis of the (de)localization transition   we are going to 
assume at times  that  we have 
\begin{equation}
\label{eq:assumefordeloc}
H'(x)\stackrel{x\searrow 0}=H'(0) -c_H x+o(x)\,,
\end{equation}
 with $c_H>0$ and 
 $H'(0):=\lim_{x\searrow 0} H'(x)$. In particular, we are going to assume \eqref{eq:assumefordeloc} only when  $H'(0)< \infty$ (we will see that  $H'(0)< \infty$ is necessary and sufficient for the existence of a delocalization transition). 
Note that if  $H''(0)<0$ exists, as limit of $H''(x)$, then  \eqref{eq:assumefordeloc} holds true 
with $c_H= \vert H''(0)\vert$. 

\medskip

We are now ready to define the non disordered model, that is the probability law $\bP_{N,h}^\Psi$
 that depends also on the real parameter $h$
 \begin{equation}
 \label{eq:mod-without-dis}
 \bP_{N, h}^\Psi (\{A\})\, :=\, \frac{ Z_{N, h}^\Psi (\{A\})}{Z_{N, h}^\Psi}\, ,
 \end{equation}
 where $A \subset \{0, \ldots, N\}$,
 \begin{equation}
 \label{eq:ZA}
 Z_{N, h}^\Psi (\{A\})\, :=\, 
 \bE \left[ \exp\left( h \sum_{j=1}^N \gd_j \right) \Psi \left( \sum_{j=1}^N \gd_j, N \right) \ind_{\tau \cap [0, N]=A} \gd_N\right]\,,
 \end{equation}
 and $ Z_{N, h}^\Psi := \sum_A Z_{N, h}^\Psi (\{A\})$, that is $ Z_{N, h}^\Psi$ coincides with the right-hand side of
 \eqref{eq:ZA} without the restriction to  $\tau \cap [0, N]=A$. We write $\{A\}$, instead of simply $A$, because  $\{A\}$
 is an elementary event and 
 $\bP_{N,h}^\Psi$ is a probability on the discrete space  $\cP (\{0,1, \ldots, N\})$, with $\cP(\cdot)$ the set of all subsets of $\cdot$.  
 Note that $\bP_{N, h}^\Psi (\{A\})=0$ unless both $0$ and $N$ are in $A$.
 
 \smallskip
 
 For the disordered version of the model  we introduce the IID sequence $(\go_n)_{n \in \bbN}$ with law $\bbP$. 
 We assume that $\gl(s):= \log \bbE[\exp(s \go_1)] < \infty$ for every $s \in \bbR$ and, without loss of generality, we set  
  $\bbE[\go_1]=0$ and $\bbE[\go_1^2]=1$. Moreover the two random sequences  $\tau $ and $\go$ are independent.
  For every realization of the disorder sequence,
   the disordered model has partition function 
  \begin{equation}
 \label{eq:Zdis}
 Z_{N, \go, \gb, h}^\Psi \, :=\, 
 \bE \left[ \exp\left(  \sum_{j=1}^N (\gb \go_j +h)\gd_j \right) \Psi \left( \sum_{j=1}^N \gd_j, N \right) \gd_N\right]\,,
 \end{equation}
 where $\gb \ge 0$. Of course, the definition of  $\bP_{N, \go, \gb,h}^\Psi $ is immediately inferred by analogy with
 \eqref{eq:mod-without-dis}.
 
 Observe now that we can write 
 \begin{equation}
 \label{eq:basic-decomp}
 Z_{N, \go, \gb, h}^\Psi \,=\, \sum_{m=1}^N
   \exp(mh)  \Psi \left(m, N \right)
  \bE \left[ \exp\left( \gb  \sum_{j=1}^N  \go_j \gd_j \right) \ind_{ \tau_m=N} \right]\,,
\end{equation}
where we have used that $\{\vert\tau \cap (0, N]\vert=m$ and $N\in \tau \}=\{\tau_m=N\}$
and if $\gb=0$, that is in the non disordered case, this expression becomes even more explicit:
\begin{equation}
 \label{eq:basic-decomp2}
 Z_{N,  h}^\Psi 
  \,=\, \sum_{m=1}^N
   \exp(mh)  \Psi \left(m, N \right)
  \bP \left(  \tau_m=N\right)
  \,.
\end{equation} 
 
\subsection{Results} 

We start with a result that shows that the free energy density exists and that it can be represented in terms of 
the \emph{ free energy density } of a pinning model in which we have fixed the number of contacts. 

\medskip

\begin{theorem}
\label{th:GF}
For every $\gb\ge 0$ and  $\rho \in [0,1]$  we have that $\bbP(\dd \go)$-a.s. 
the limit
\begin{equation}
\label{eq:G}
\limtwo{N \to \infty:}{m/N \to \rho}
\frac1N \log  \bE \left[ \exp\left( \gb  \sum_{j=1}^N  \go_j \gd_j \right) \ind_{ \tau_m=N}\right]\,=:\, 
\tg(\gb, \rho)\,,
\end{equation}
exists and $\tg(\gb, \rho)\in \bbR$ is non random. Moreover $\tg(\gb, \cdot)$ is concave, $\lim_{\rho \searrow 0} \tg(\gb, \rho)=
 \tg(\gb, 0)=0$, $\lim_{\rho \nearrow 1} \tg(\gb, \rho)=\tg(\gb, 1)=\log K(1)$ and  $\lim_{\rho \nearrow 1} \partial_\rho \tg(\gb, \rho)=-\infty$.

Also the limit
\begin{equation}
\label{eq:FH}
\lim_{N \to \infty} \frac 1N \log Z^\Psi_{N , \go, \gb, h}\, =:\, \tf_H(\gb, h)\, ,
\end{equation}
exists for every $\gb\ge 0 $ and $h \in \bbR$,  $\bbP(\dd \go)$-a.s. and in $\bbL^1$. $\tf_H(\gb, h) \in \bbR$ is non random,
$\tf_H(\gb, \cdot)$ is non decreasing, convex and we have the conjugate variational formulas
\begin{equation}
\label{eq:FHG}
\tf_H(\gb, h)\,=\, \sup_{\rho \in [0,1]} \left( h \rho +H(\rho) + \tg(\gb, \rho)\right) \ \text{ and } \
\tg(\gb, \rho)\,=\, \inf_{h \in \bbR} \left( \tf_H(\gb, h)-\rho h -H(\rho) \right)\,.
\end{equation} 
\end{theorem}

\medskip

Of course $\tf_H(\gb, h)$ is the free energy (density) of the model defined by \eqref{eq:Zdis}.
We point out that at this stage that $ \partial_\rho \tg(\gb, \rho)$ should be interpreted as 
the limit of the incremental ratio from the left, or from the right: they both exist by concavity. 
We will see in Proposition~\ref{th:Gsmooth}  that $\tg(\gb, \cdot)$ is $C^\infty$ if $\gb>0$
and $\tg(0,\cdot)$ is real analytic, except possibly in one point, in which it is in any case at least $C^1$, see   Proposition~\ref{th:propG}.

\medskip
 
It is worth pointing out that if $\Psi(m, N)=1$ for every $m$ and $N$ then  the model coincides with the well known 
disordered pinning model:
\begin{equation}
 \label{eq:ZPS}
 Z_{N, \go, \gb, h} \, :=\, 
 \bE \left[ \exp\left(  \sum_{j=1}^N (\gb \go_j +h)\gd_j \right)  \gd_N\right]\,,
 \end{equation}
 and the corresponding free energy is denoted by $\tf(\gb, h)$.
Of course also the case $\Psi(m, N)=\exp(am+b)$, $a$ and $b$ real constants, corresponds to trivial modifications of the pinning model.
As a matter of fact, Theorem~\ref{th:GF} is telling us that,
whenever $H$ of Definition~\ref{def:Psi} is trivial (i.e., affine), we are dealing with a model with free energy 
that coincides, up to an additive  constant and a shift in $h$, with the free energy of a pinning model. 
   
\medskip

In order to better appreciate the results let us consider first the case $\gb=0$: when, like here, there is no risk of confusion, we drop the dependence on $\gb(=0)$, that is 
we write  $\tf_H( h)$ for  $\tf_H(0, h)$, etc$\ldots$. 

The first $\gb=0$ result says that $h \mapsto \tf_H(h)$ may have up to two singularity
 points if $\ga>1$. Otherwise, that is if $\ga \in (0,1]$, it has at most one singularity. We introduce 
\begin{equation}
\rho_c\, :=\, \frac 1{\bE[\eta]}\, ,
\end{equation}
so $\rho_c=0$ if $\ga \in (0,1]$ and $\rho_c\in (0,1)$ if $\ga >1$.
Moreover, given a model (that is given $\Psi$, hence $H$), 
 we set 
 \begin{equation}
 h_c^H\,:=\, -H'(0)\in [-\infty, \infty) \ \ \ \ \textrm{ and } \ \ \ \ h_b\, :=\, -H'(\rho_c)\, .
 \end{equation}
  Of course $h_c^H\le h_b$. Moreover $h_c^H=h_b$ if   $\rho_c=0$, 
  and we can replace \emph{if} with \emph{if and only if} when $H(\cdot)$ is not trivial.

\medskip

\begin{theorem}
\label{th:Fbeta0}
The function $h \mapsto \tf_H(h)$ is real analytic for $h\in \bbR\setminus \{h_c^H, h_b\}$ and it is not real analytic in $h_c^H$ nor in $h_b$. Moreover:
\begin{enumerate}
\item if $h_c^H\in \bbR$ (and regardless of $h_c^H<h_b$ or $h_c^H=h_b$) 
$\tf_H(h)= H(0)\in \bbR$ for $h\le h_c^H$ and $\tf_H(h)> H(0)$ for $h>h_c^H$
with the $h \searrow h_c^H$ asymptotic behaviors (assuming \eqref{eq:assumefordeloc}):
\begin{itemize}
\item for $\ga>1$ we have $\tf_H(\rho)-H(0)\sim (h-h_c^H)^2/(2c_H)$; 
\item for $\ga\in (0,1]$ we still have $\tf_H(\rho)-H(0)\sim (h-h_c^H)^2/(2c_H)$ for $\ga >1/2$ and the same
 is true for $\ga=1/2$ but with a  prefactor  smaller than $1/(2c_H)$. For $\ga \in (0, 1/2)$ we have
 $\tf_H(\rho)-H(0)\sim c(h-h_c^H)^{1/\ga}$ for a suitable $c>0$;
\end{itemize}
\item  if $h_c^H<h_b$ (which implies $\ga >1$) then 
the function $h\mapsto \tf_H^{\textrm{reg}}(h)$ defined by 
\begin{equation}
\label{eq:Fbeta0-1}
\tf_H^{\textrm{reg}}(h)\, :=\, \sup_{\rho \in [0,1]} \left( h \rho +H(\rho)\right)\, ,
\end{equation}
 is real analytic in 
$(h_c^H, 1)$ and, with $\kappa:= \max(\ga/(\ga-1), 2)$ and for a suitable constant $c>0$ (that depends in  particular on $\ga$),  we have (except for the case $\ga=2$)
\begin{equation}
\label{eq:Fbeta0-2}
 \tf_H(h)- 
 \tf_H^{\textrm{reg}}(h) \stackrel{h \searrow h_b}\sim -c (h-h_b)^{\kappa} \, ,
\end{equation}
while $ \tf_H(h)=
 \tf_H^{\textrm{reg}}(h) $ for $h<h_b$.
 If $\ga=2$ the right-hand side in \eqref{eq:Fbeta0-2} is replaced by  $-c (h-h_b)^{2}/ \log(1/ (h-h_b))$.
\end{enumerate}
\end{theorem}

\medskip

Theorem~\ref{th:Fbeta0}(2) is established in \cite{cf:bar1,cf:bar2} for specific choices of $\Psi$, see Appendix~\ref{sec:circDNA}, and for $K(n)$ equal, not simply asymptotically equivalent, to $C_K/{n^{1+ \ga}}$. The approach in in \cite{cf:bar1,cf:bar2}
exploits the expression for the Mellin transform of $K(\cdot)$ in terms of special functions  and by doing the asymptotic analysis 
via identification of singularities in the complex plain. Our analysis is more general and substantially simpler. 

\medskip

A direct consequence of Theorem~\ref{th:Fbeta0} is that $\tf_H(\cdot)$ is differentiable. In fact,  the two singularity loci are  $h_c^H$ and $h_b$ and
\begin{itemize}
\item the quadratic behavior at $h_c^H$, proven assuming \eqref{eq:assumefordeloc},  of course yields differentiability, but we take this occasion to stress that a first order transition, i.e. discontinuous $ \tf'_H(\cdot)$,  happens only if $H(\cdot)$ is trivial: a look at the proof suffices to conclude that the contact fraction is continuous at $h_c^H$
as soon as $H(\cdot)$ is strictly concave, i.e. without assuming \eqref{eq:assumefordeloc};  
\item when $h_b>h_c^H$ the critical exponent $\kappa$ is larger than one, and, again, this yields differentiability.  
\end{itemize}

\medskip

The transition at $h_c^H$ is a delocalization/localization transition: 
in fact $\tf_H'(h)=0$ for $h< h_c^H$ and $\tf_H'(h)>0$ for $h>h_c^H$ and 
$\tf_H'(h)$ coincides  with $\lim_N \bE_{N, h}^\Psi[\sum_{j=1}^N \gd_j]/N$, which is  the contact density. Note that for this transition we assume 
\eqref{eq:assumefordeloc}, so $H(\cdot)$ is non trivial, and the critical exponent coincides with the critical exponent of the standard pinning model only for $\ga \le 1/2$.
For $\ga>1/2$ the critical exponent of the $\gb=0$ pinning model is $\max(1/\ga,1)$ \cite[Ch.~2]{cf:GB}.

\medskip

\begin{rem}
\label{rem:de-loc} 
A more direct view of the delocalization/localization transition can be taken, without assuming $\gb=0$, by noticing that
$\tf_H(\gb, h) \ge H(0)$. This fact  is straightforward 
under the stronger condition $\liminf_N (\log Q(1, N))/N\ge 0$ because  it suffices to restrict the partition function to the event $\tau_1=N$
(that is, the only  contact point is in $N$) and we obtain for every $b>0$ and $N$ sufficiently large
\begin{equation}
 Z^\Psi_{N, \go ,\gb ,h} \, \ge \, 
 \exp \left(N(H(1/N)-b)\right) \exp(\gb \go_N +h) 
 \bP( \tau_1=N)
 \, .
\end{equation}
So $\lim_N (1/N) \log Z^\Psi_{N, \go, \gb ,h}\ge H(0)$ a.s. and we are done. For a proof without the additional  assumption see Proposition~\ref{th:de-loc2}. 
\end{rem}

\medskip

The transition at $h_b$, when $h_b>h_c^H$, corresponds in  physical terms to the appearing/disappearing of  a \emph{condensation} segment or of a  \emph{macroscopic loop}. The underlying phenomenon is well known also in the probability literature and it is called \emph{big-jump} (see \cite{cf:AL,cf:DDS08} and references therein). 
In order to make this precise  we introduce for every $N$ the order statistics of the $\eta$ sequence \emph{up to $N$}, that is the order statistics of $\eta_1, \eta_2, \ldots, \eta_{\vert \tau \cap (0, N]\vert}$, for which we use the notation   
$\eta_{1,N} \ge \eta_{2,N} \ge \ldots \ge \eta_{\vert \tau \cap (0, N]\vert,N}$. Note that this order statistics is empty 
if $\tau_1>N$, which never happens  because we always work with $N \in \tau$. On the other hand, what may  happen is that   the sequence contains only one element, that is 
$\eta_{1,N}=\tau_1=N$,  and in this case we set $ \eta_{2,N}=0$.

\medskip

\begin{theorem}
\label{th:paths0}
For every $h\neq h_c^H$  we have that in $\bP^\Psi_{N, h}$-probability
\begin{equation}
\label{eq:paths0}
\lim_{N\to \infty} \frac{\left \vert \tau \cap (0, N) \right\vert}N\, =\, \rho_h :=\, \tf_H'(h)
\,,
\end{equation}
and when $\rho_c>0$ and $h \neq h_b$
\begin{equation}
\label{eq:paths0-1}
\lim_{N\to \infty} \frac{\eta_{1,N}}N\, =\, \left ( 1- \frac {\rho_h} {\rho_c} \right)_+,
\ \ \ \ \ \ 
\ \ \ \ \ 
\lim_{N\to \infty} \frac{ \eta_{2,N}}N\, =\, 0\,.
\end{equation}
If $\rho_c=0$ and $h> h_c^H=h_b$ we have that $\lim_N \eta_{1,N}/N=0$. 
\end{theorem}

\medskip

The only reason to require $h \neq h_b$ is to keep proofs concise: the statement holds without this requirement  \cite[Ch.~4]{cf:thesis-BH}. The same is true for $h\neq h_c^H$ and  \eqref{eq:paths0}, if one takes care of excluding the cases in which 
$ \tf_H'(h_c^H)$ does not exist: note  that, by Theorem~\ref{th:Fbeta0}(1), $ \tf_H'(h_c^H)$ exists and it is equal to zero for $H(\cdot)$ non trivial.

Informally stated, Theorem~\ref{th:paths0} is spelling out the standard fact that $\tf_H'(h)$ is the contact fraction and  that the 
largest loop $\eta_{1,N}$ encompasses  essentially all the system in the delocalized regime $h< h_c^H$ (but only if $\rho_c>0$! See Remark~\ref{rem:phyp}) and it is instead macroscopically negligible (i.e., $\eta_{1,N}=o(N)$) if $h>h_b$. But the key point for us is that when $h_c^H<h_b$, for 
$h\in (h_c^H, h_b)$  the 
largest loop $\eta_{1,N}$, normalized by dividing by $N$, is asymptotically of size $(1-\rho_h/ \rho_c) \in (0,1)$. Moreover, cf. \eqref{eq:paths0-1}, all other loops are macroscopically negligible.  

Moreover, as we have seen, when $h_c^H<h_b$ then 
$\rho_h=\tf_H'(h)$ is continuous both at $h_c^H$ and at $h_b$ and this implies that the normalized large loop size 
behaves continuously at these transitions (in the first case it goes to one, in the second one it vanishes). 

\medskip

\begin{rem}
\label{rem:phyp}
The reader may be surprised by the lack of a full path delocalization result for $\rho_c=0$, i.e. $\ga \in (0, 1]$, like for $\ga >1$.
We are convinced that this  cannot be obtained with our assumptions on $Q(\cdot, \cdot)$: see the control from below
 in \eqref{eq:defPsi1} of Def.~\ref{def:Psi}. Note in fact that, for example,    we can choose $Q(m, N)= \exp(-N^2)$    for $m\le N/ \log N$ and 
 this forces the presence of at least about $N/\log N$ contacts: forcing them to be close to the boundary is very expensive in probability terms.
\end{rem}

\medskip

Theorem~\ref{th:paths0} can be improved in a number of ways, notably the largest loop for $h>h_b$ is $O(\log N)$, while 
the second largest loop for $h \in (h_c^H, h_b)$ has a power law scaling and for $h< h_c^H$ is $O(1)$. These issues are developed in \cite[Ch.~4]{cf:thesis-BH}, along with a detailed analysis of the critical cases. Our focus is on the effect of the disorder on the system and Theorems~\ref{th:Fbeta0} and \ref{th:paths0} are sufficient for this purpose. 

In fact, the main point of our work is that for $\gb>0$ the big-jump phenomenon disappears, and this is what we present next, along with an analysis of the effect of the disorder on the (de)localization transition. 

\medskip

Recall Remark~\ref{rem:de-loc} and set $h_c^H(\gb):= \inf\{h:\, \tf_H(\gb, h)>H(0)\}$. Of course we have
$h_c^H(\gb):= \sup\{h:\, \tf_H(\gb, h)=H(0)\}$ and $h_c^H=h_c^H(0)$. If $H(\cdot) \equiv 0$ we use $h_c(\gb)$ for $h_c^H(\gb)$, 
in parallel with the use of $\tf (\gb, h)$ for 
$\tf _H(\gb, h)$. 
Much work has been done on identifying as precisely as possible $h_c(\gb)$: 
a through review of the literature is in Section~\ref{sec:sketchandmore}. Here we only anticipate that 
$h_c(\gb)\in [-\gl (\gb), 0]$ for every $\gb \ge 0$.   

\medskip

\begin{theorem}
\label{th:beta>0}
For $\gb>0$ we have that 
\begin{enumerate}
\item
$h \mapsto \tf_H(\gb, h)$ is $C^\infty$ for $h\in (h_c^H(\gb), \infty)$ and
$h_c^H(\gb)= h_c(\gb) -H'(0)$, so $h_c^H(\gb)>-\infty$ if and only if $H'(0)<\infty$ (like for $\gb=0$).
If $H'(0)<\infty$ (without assuming \eqref{eq:assumefordeloc})
there exists $C_\gb$  and $\gD >0$
such that if  $h-h^H_c(\gb)\in [0, \gD]$
\begin{equation}
\label{eq:smoothingH}
\tf_H(\gb,h) \, \le \, H(0) + C_\gb(h-h_c(\gb))^2\, ,
\end{equation}
and, assuming\eqref{eq:assumefordeloc}, for  $\ga \in (0,1/2)$ and   $\gb\in [0, \gb_\ga]$, 
for a suitable choice of $\gb_\ga>0$,  we have the sharper result
\begin{equation}
\label{eq:sharpH}
\tf_H(\gb, h_c^H(\gb) +y)-  H(0)  \stackrel{y\searrow 0}\sim
\tf_H(0, h_c^H(0) +y)-  H(0) \, .
\end{equation}
 \item for every $h$ 
 we have that $\bbP(\dd \go)$-a.s.
in $\bP^\Psi_{N,\go,\gb, h}$-probability 
\begin{equation}
\label{eq:beta>0contact}
\lim_{N\to \infty} \frac{\left \vert \tau \cap (0, N) \right\vert}N \, =\,  \partial_h\tf_H (\gb,h)\, ,
\end{equation}
 and for $h>h_c^H(\gb)$
\begin{equation}
\label{eq:beta>0}
\lim_{N\to \infty} \frac{\eta_{1,N}}N\, =\, 0\,.
\end{equation}
\end{enumerate}
\end{theorem}

\medskip

A number of comments are in order:
\smallskip

\begin{enumerate}
\item $h_c^H(\gb)$ may be equal to $-\infty$, but otherwise   $\tf_H(\gb, \cdot)$ is not  analytic at $h_c^H(\gb)$, which is therefore a critical point marking the transition from zero contact density (delocalized regime) to positive contact density (localized regime);  
\item from the proof we see that $C_\gb$ can be chosen independent of $\gb$ if we assume \eqref{eq:assumefordeloc} (see Remark~\ref{rem:smoothing});  
\item
the finite order big-jump transition at $h_b$ has disappeared, but the  $C^\infty$ regularity estimate on the free energy  leaves open the possibility of an infinite order transition;
\item nevertheless, \eqref{eq:beta>0} tells us that the loops in the localized regime do not have macroscopic size, so the large loop phenomenon is washed out by the disorder;
\item we have decided to leave aside the delicate  analysis of the path behavior in the delocalized phase: we certainly expect that results like in  \cite{cf:AZ14,cf:GT1} can be adapted, but only under stronger conditions on $\Psi(m, N)$ (and the problem is already present for $\gb=0$, see Remark~\ref{rem:phyp}).
\end{enumerate}
\medskip


Two important remarks:

\medskip

\begin{rem}
\label{rem:smoothing}
The proof of \eqref{eq:smoothingH} exploits the smoothing inequality \cite{cf:CH13,cf:GTprl,cf:GT06} that we recall in \eqref{eq:smoothing1} below, but only in part
because the result holds as soon as $H''(\rho)$ stays bounded away from $0$ for $\rho$ close to zero, and in particular when \eqref{eq:assumefordeloc} holds.  And, in view of the $\gb=0$ results in 
Theorem~\ref{th:Fbeta0}(1),  \eqref{eq:smoothingH} does not establish a smoothing phenomenon. Disorder relevance is certainly expected and it would follow from what is expected to hold  for the disordered pinning model (that is, an infinite order transition for $\ga\ge 1/2$, see Remark~\ref{rem:open}). In our model however we can see smoothing for $\ga>1/2$ if we do not assume 
\eqref{eq:assumefordeloc}. Notably if we assume for example that $H'(\rho)-H'(0) \sim -c \rho ^\gamma$ with a $\gamma >1$ and $c>0$: 
\eqref{eq:smoothingH} holds, but Theorem~\ref{th:Fbeta0}(1) changes and $\tf_H(0, h_c(0)+ \gd)$ becomes equivalent to $\gd^{\max(1+ 1/\gamma,1/ \ga)}$ times a positive constant. In this case   \eqref{eq:smoothingH} does establish a smoothing phenomenon and disorder relevance. Finally, \eqref{eq:sharpH}  establishes
disorder irrelevance for $\ga <1/2$.
\end{rem}

\medskip

\begin{rem}
\label{rem:Harris}
The Harris criterion  is applied in \cite[Sec.~IV]{cf:bar2} to the big-jump transition and the claim is that disorder is irrelevant for this transition for $\ga \in (1, 2]$, while for $\ga>2$ disorder is ``marginal", i.e. at the boundary  between irrelevance  and relevance.
 This is in contrast with Theorem~\ref{th:beta>0} which proves relevance of the disorder for every $\ga>1$.  It would be of course very interesting to understand what is happening in the Harris' perspective. We take this occasion to point out that the ``instability" of the big-jump transition under the effect of disorder has been observed also in \cite{cf:GCS,cf:molino}. In \cite{cf:GCS,cf:molino}  the disorder is introduced in such a way that the  renewal  structure is preserved and explicit computations can be performed. In our case there is no such structure and our results follow from the smoothing inequality  bound for the standard pinning case \cite{cf:CH13,cf:GTprl,cf:GT06}. While we believe that our disorder relevance result for big-jump transitions should hold in greater generality, our approach does not generalize in an evident way, notably not to the tightly related gPS model mentioned in Section~\ref{sec:preintro}. 
 The contribution \cite{cf:BGK2}  deals with  the disorder (ir)relevance issue in the gPS model, but only for the localization transition.
\end{rem}

\subsubsection{Organization of the paper}
\begin{itemize}[leftmargin=5.5mm]
\item
In Section~\ref{sec:sketchandmore} 
we present the main ideas on how we deal with the disorder and we provide a proof of Theorem~\ref{th:beta>0}(1), relying on the variational formulas of 
Theorem~\ref{th:GF} and on the uniform strict convexity bound of Theorem~\ref{th:unifconvex}.
\item
In Section~\ref{sec:3} we provide a full analysis of  $\tf_H(0, h)$. In particular, this section contains the proof of Theorem~\ref{th:Fbeta0}.
\item
In Section~\ref{sec:path-beta0} we analyse the trajectories of the process for $\gb=0$ (proof of Theorem~\ref{th:paths0}).
\item
In Section~\ref{sec:pathsbeta>0} we complete the proof of Theorem~\ref{th:beta>0}, by proving part (2) that concerns the trajectories: no big-jump for $\gb>0$. 
\item
In Section~\ref{sec:var} we take care of the free energy existence issues and of the variational formulas (proof of Theorem~\ref{th:GF}).
\item
In Appendix~\ref{sec:circDNA} we explain how   the circular DNA case \cite{cf:bar1,cf:bar2} fits in our framework and in Appendix~\ref{sec:unifconvex}
we  prove that $\partial_h^2\tf(\gb,h)>0$ for every $h>h_c(\gb)$ and we complete Remark~\ref{rem:de-loc}.
\end{itemize}

\section{Exploiting the Legendre transform and the key role of the standard pinning  model}
\label{sec:sketchandmore}

Let us start by pointing out the direct consequence of \eqref{eq:FHG}
\begin{equation}
\label{eq:FG}
\tf(\gb, h)\,=\, \sup_{\rho \in [0,1]} \left( h \rho  + \tg(\gb, \rho)\right) \ \text{ and } \
\tg(\gb, \rho)\,=\, \inf_{h \in \bbR} \left( \tf(\gb, h)-\rho h  \right)\,.
\end{equation} 
The strategy we employ is to obtain information 
on $\tg(\gb, \rho)$, defined  in \eqref{eq:G}, via the second formula in \eqref{eq:FG} and what we know  on $\tf(\gb, h)$. 
We start therefore by collecting here the relevant known results on $\tf(\gb, h)$: some of these results are straightforward, but most of them   are the outcome of the work of several contributors.

\medskip

\begin{enumerate}[start=1,label={({\bfseries P\arabic*})}]
\item \label{P:basic}
 Basic convexity and monotonicity properties, together with some relatively standard bounds,  show that $\tf(\gb, \cdot)$ is convex, it is equal to $0$ for $h \le h_c(\gb)$ and it is positive and increasing for $h>h_c(\gb)$:   $h_c(\gb)\le 0$ and for more on its value  see \ref{P:critpoint}. 
\item  \label{P:smooth} $h \mapsto \tf(\gb, h)$ is $C^\infty$ for $h>h_c(\gb)$ \cite[th.~2.1]{cf:GTalea} (in \cite{cf:GTalea} a concentration condition is required on the law of $\go_1$, but this is not used in the proof of Theorem~2.1) 
 and it is analytical if $\gb=0$ for $h>h_c(0)$, see e.g.   either \cite[Ch.~2]{cf:GB} or \cite[Ch.~2]{cf:G}. Of course it is also analytical for $h< h_c(\gb)$ and $h_c(\gb)$ is a non analyticity point.  We add that for every $\gb\ge 0$ it is straightforward to show that  
$\lim_{h \to \infty}\partial \tf_h(\gb, h) =1$.
\item \label{P:beta0}
For $\gb=0$ the model is solvable \cite{cf:Fisher,cf:GB,cf:G}: we have already pointed out that $h_c(0)=0$, but the sharp behavior of the free energy and its derivatives at criticality is available too. That is, for $\ga\in (0, 1)$ there exists $c_\ga>0$ such that
\begin{equation}
\label{eq:beta0crit}
\tf(0, h) \stackrel{h \searrow 0}\sim c_\ga h^{1/\ga}\, , 
\end{equation} 
and \eqref{eq:beta0crit} holds also if we differentiate $k\in \bbN$ times both sides. If $\ga >1$ instead 
$\tf(0, h)\sim h/\bE[\eta]$ and this statement can be differentiated once. If $\ga =1$ instead $\tf(0, h)\sim C_K h/\log(1/h)$ (one differentiation allowed). 
 These results imply the rougher statement
\begin{equation}
\label{eq:exponent}
\frac{\log \tf(0, h)}{\log h} 
\stackrel{h \searrow 0}\sim  \max \left( 1 , \frac 1 \ga\right)\, .
\end{equation}
Notably, the transition is of first order if  $\ga>1$ and it is of higher order if $\ga \in (0,1]$: 
\begin{equation}
\lim_{h\searrow 0}\partial_h\tf(0, h)\, =\, \frac 1 {\bE[\eta]}\, 
\begin{cases}  >0 & \text{ if } \ga >1\, , \\
=0 & \text{ if } \ga \in (0,1]\,.
\end{cases}
\end{equation}
\item \label{P:critpoint}
For $\gb>0$ we have $- \gl(\gb) \le h_c(\gb) < h_c(0)=0 $ (see \cite[Ch.~3]{cf:G} and  \cite[Ch.~4]{cf:GB};
see \cite{cf:AS}  and  \cite[Section.~5.2]{cf:GB}  for the strict inequality). Moreover $h_c(\gb)>-\gl(\gb)$ for $\ga \ge 1/2$ \cite{cf:AZ09,cf:BL16,cf:CTT,cf:DGLT09}, but 
 $h_c(\gb)=-\gl(\gb)$ for $\ga \in (0, 1/2)$ and $\gb\le \gb_\ga$, for a suitable $\gb_\ga>0$ \cite{cf:A08,cf:hubert,cf:fabio}.
\item \label{P:smoothing1}
For $\gb>0$ there exists $c_\gb>0$ and $\gD_0>0$ such that for every $\gD \in (0, \gD_0]$
\begin{equation}
\label{eq:smoothing1}
0 \,<\, \tf(\gb, h_c(\gb)+\gD)-\tf(\gb, h_c(\gb)) \,\le \, c_\gb \gD^2\, .
\end{equation}
The lower bound in \eqref{eq:smoothing1} is trivial, the upper bound is the smoothing inequality \cite{cf:CH13,cf:GTprl,cf:GT06} . We stress that  \eqref{eq:smoothing1} directly implies that, regardless of the value of $\ga$, for $\gb>0$ we have 
\begin{equation}
\label{eq:smoothing2}
\lim_{h\searrow h_c(\gb)}\partial_h\tf(\gb, h)\,=\, 0\,.
\end{equation}
\item \label{P:irrelevant} If $\ga \in (0,1/2)$  ( and $\gb \le \gb_\ga$, see \ref{P:critpoint}) we have \cite{cf:A08,cf:GT09,cf:hubert,cf:fabio}
\begin{equation}
 \label{eq:irrelevant}
\tf(\gb, h_c(\gb)+y)\stackrel{y \searrow 0}\sim \tf(0, y)\,.
\end{equation}
\end{enumerate}

\medskip

\begin{rem}
\label{rem:open}
The truly open problem for the disordered pinning model, and, as a matter of fact, for every disorder relevant model, is what is the precise critical behavior when disorder is relevant, see  \cite[Ch.~5]{cf:G} for a discussion and references. For the pinning model in the relevant disorder regime  it is now expected that the transition becomes of infinite order. One of the main reason is that the model is expected to be in the strong or infinite disorder universality class \cite{cf:DM80,cf:dfisher,cf:IM}, see also the more recent contribution \cite{cf:DR}. In this line there have been also some mathematical progress \cite{cf:BGL,cf:CDHLS}, but they do not impact directly the pinning model.  
\end{rem}

\subsection{Legendre transform viewpoint on homogeneous pinning}
Recalling \eqref{eq:basic-decomp2} and \eqref{eq:G} we see that 
$\tg(\rho)=\tg(0, \rho)$ has a very simple expression:
\begin{equation}
\tg(\rho)\, =\, \limtwo{N\to \infty:}{m/N\to \rho} \frac 1N \log \bP\left( \tau_m =N\right)\, ,
\end{equation}
and arbitrarily precise estimates on $\tg(\cdot)$ can be obtained, see in particular Proposition~\ref{th:propG} that is  resumed in part in  Figure~\ref{fig:alpha1} and in
Figure~\ref{fig:alpha2}, and their captions.
The behavior of $\tf(\cdot)=\tf(0, \cdot)$ can then be extracted from $\tg(\cdot)$  via \eqref{eq:FG}. In particular, 
the non analytic behavior of $\tg(\cdot)$ at $\rho_c>0$ yields a jump of size $\rho_c$  in $\tf'(\cdot)$ at $h_c(0)=0$.
The jump disappears if $\rho_c=0$. Moreover, these implications can be reversed, and the behavior of $\tg(\cdot)$
can be inferred from the one of $\tf(\cdot)$.

\begin{figure}[htbp]
\centering
\includegraphics[width=14 cm]{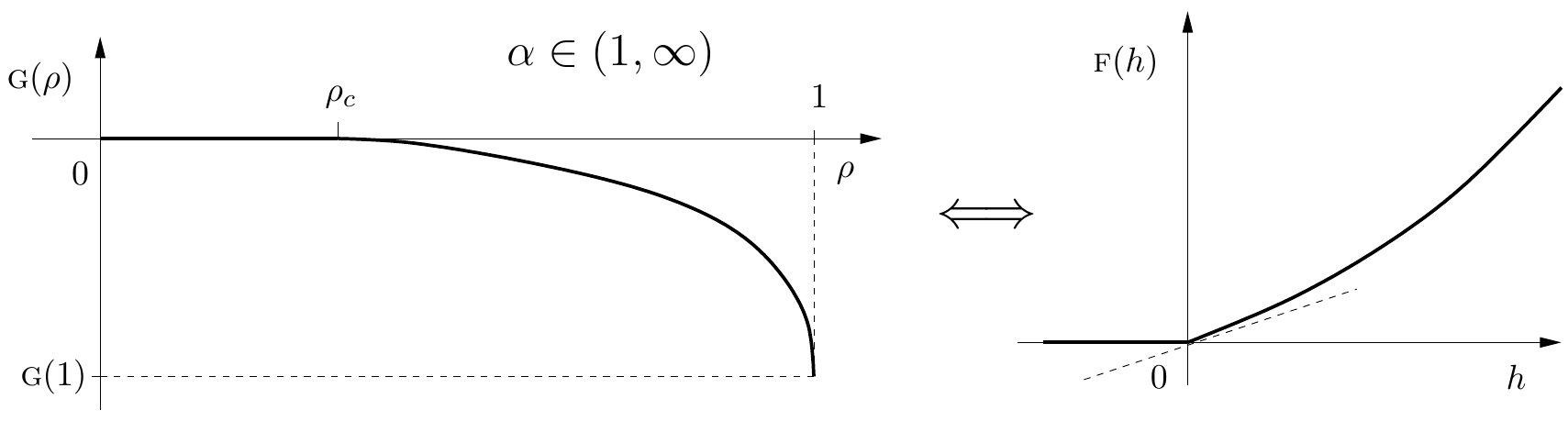}
\caption{\label{fig:alpha1} 
The figure illustrates the (Legendre transform) link between $\tg(\cdot)$ and $\tf(\cdot)$ for $\ga>1$ (and $\gb=0$).
We stress here that the flat portion of $\tg(\cdot)$, that is 
 $\tg(\rho)=0$ for $\rho \in [0, \rho_c]$ with $\rho_c=1/\sum_n nK(n)$, has a direct counterpart the first order localization transition for the pinning model: quantitatively, 
 $\lim_{h \searrow 0}\tf'(h)= \rho_c$. Other features stressed in the graph are (1)  the fact that $\tg(1)=-\log K(1)$ is finite but 
 $\lim_{\rho \nearrow 1}\tg'(\rho)= -\infty$ and (2) that $\lim_{\rho \to \rho_c}\tg'(\rho)=0$. 
 These two features are going to be central for our generalized pinning model: feature (1)
 forces the optimizing contact density to be smaller than one and feature (2) makes the big-jump transition of order two or larger.
}
\end{figure}

The non analyticity  at $\rho=\rho_c$
can be viewed as a phase transition: in fact, $\tg(h)$ capture the exponential asymptotic behavior of $\bP\left( \tau_m =N\right)$, for $m/N \sim \rho$,
and we can view $\bP\left( \tau_m =N\right)= \bE[ \ind_{\tau_m =N}]$ as the partition function of the model which is just the renewal 
conditioned to $\tau_m =N$ (in Section~\ref{sec:path-beta0} this probability will be denoted by $\bQ_{N, m}$).
The trajectories  $\bQ_{N, m}$ are substantially different when $m/N \sim \rho$ is below or above $\rho_c$ and the phenomenon is known in 
probability as the 
\emph{big-jump} phenomenon: if $\rho< \rho_c$ a single large  excursion takes care of the anomalously low contact density, 
in fact the \emph{typical} contact density for the renewal is $\rho_c$ (this is very well known Renewal Theorem, see \cite[App.~A] {cf:G} and references therein).
Instead  the system constrained to a contact density  $\rho >\rho_c$ globally modifies itself to accomodate more excursions. 
See Figure~\ref{fig:bigjump} for a visual explanation: Proposition~\ref{th:muNm} is a mathematical presentation of the big-jump transition.

\begin{figure}[htbp]
\centering
\includegraphics[width=14 cm]{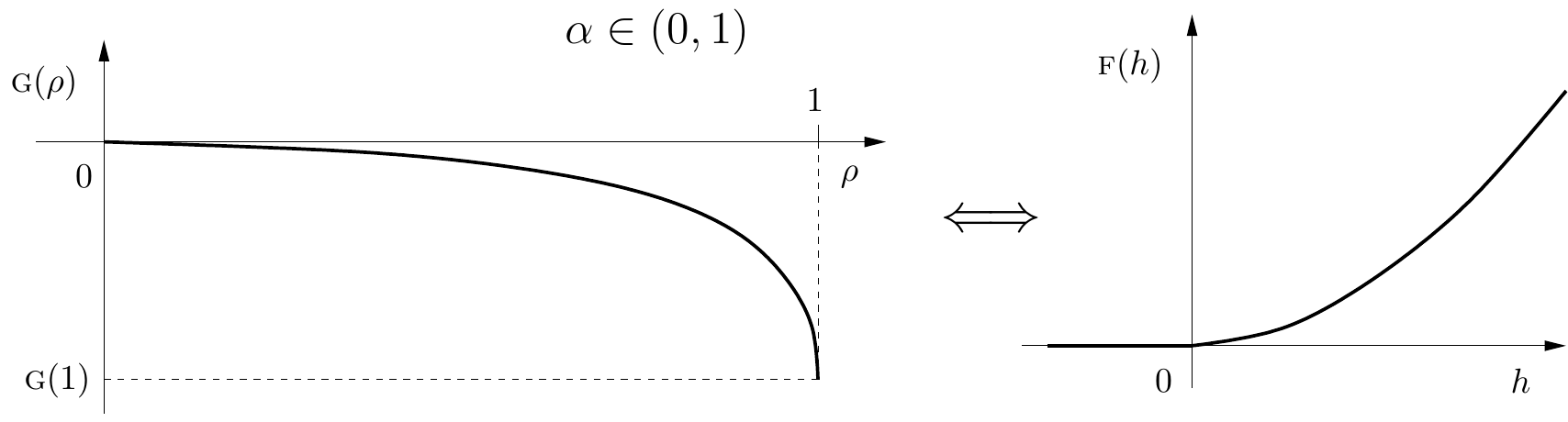}
\caption{\label{fig:alpha2} 
The figure illustrates the (Legendre transform) link between $\tg(\cdot)$ and $\tf(\cdot)$ for $\ga \in (0,1)$ (and $\gb=0$).
This time $\tg(\cdot)$ is analytic over all the domain. The corresponding behavior of $\tf(\cdot)$ is on the right and the difference 
with the case $\ga>1$ is that $\tf'(\cdot)$ exists also at the origin. In fact, the smaller $\ga$ is, the more $\tf(\cdot)$ is regular at the origin. 
But the most prominent  fact is that $\rho_c=0$ is equivalent $\tf'(\cdot)$ being $C^1$ in $0$.
}
\end{figure}

\begin{figure}[htbp]
\centering
\includegraphics[width=14 cm]{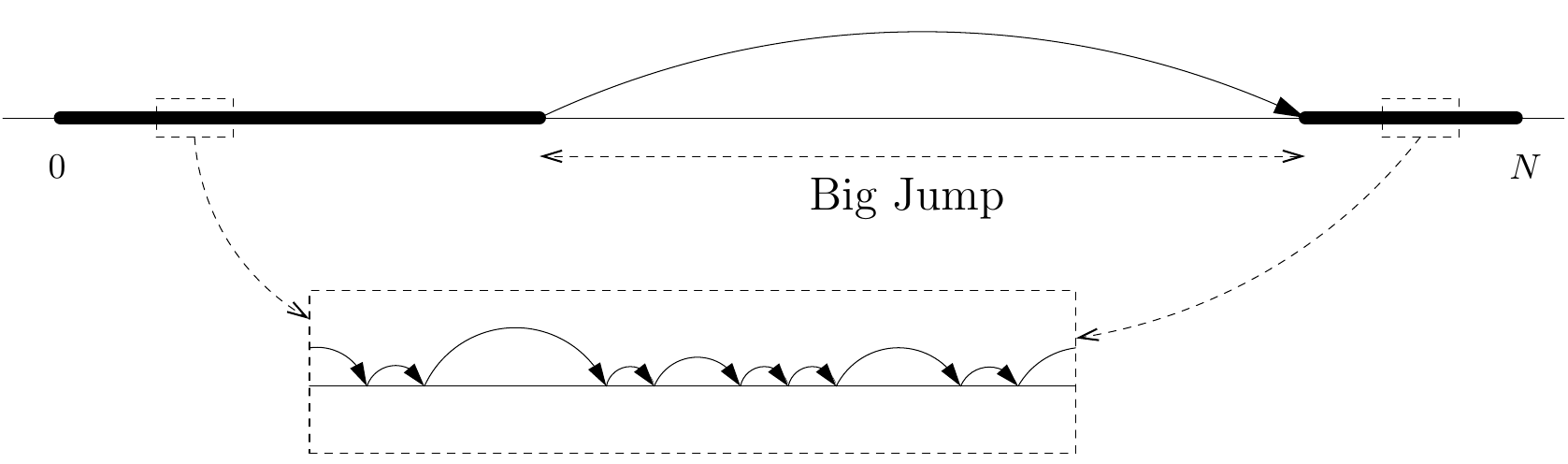}
\caption{\label{fig:bigjump} 
The big-jump phenomenon that happens when we condition the renewal to have $m$ contacts between before $N$,
with $m/N\sim \rho$ smaller than the typical value $\rho_c$. The system behaves typically, so with contact density $\rho_c=1/\bE[\eta]$, and compensates for   the low global contact density by making a \emph{big-jump} of length 
$\sim (1-\rho/ \rho_c) N$, randomly (uniformly) placed in the interval.
}
\end{figure}

\subsection{Legendre transform viewpoint on disordered pinning} 
The key point here is simply that \eqref{eq:smoothing2} is telling us that there are no longer two scenarios, but only one:
qualitatively, the one of $\gb=0$ and $\ga \in (0,1]$.  
Here is a central statement for our analysis:

\medskip

\begin{proposition}
\label{th:Gsmooth}
For $\gb>0$ we have that $\rho \mapsto \tg(\gb,\rho)$ is $C^\infty$ in the interior of its domain of definition, that is for $\rho \in (0,1)$. Moreover $\lim_{\rho \searrow0} \tg  (\gb,\rho)=0$ and 
$\lim_{\rho \searrow0} \partial_h \tg   (\gb,\rho)=-h_c(\rho)$. Finally, $ \partial^2_h \tg  (\gb,\rho)<0$ for every $\rho \in (0,1)$ and
there exists $c>0$ (depending on $\gb$) such that for every $\rho$
\begin{equation}
\label{eq:Gsmooth}
\tg(\gb, \rho)\, \le \, -h_c(\gb) \rho -c \rho^2\, .
\end{equation}
For $\ga \in (0,1/2)$ 
and $\gb \in [0, \gb_\ga]$ ($\gb_\ga$ given in  \ref{P:critpoint}) we have  that
\begin{equation}
\label{eq:Gbetarho0-3}
\tg(\gb, \rho) + h_c(\gb) \rho \stackrel{\rho \searrow 0}\sim \tg(0, \rho)\, .
\end{equation}
\end{proposition}

\medskip

\begin{proof}
Fix $\gb>0$.
By the second identity in \eqref{eq:FG} for every $\rho \in (0,1)$
\begin{equation}
\label{eq:int54s}
\tg(\gb, \rho)\,=\,  \tf(\gb, h_\rho)-\rho h_\rho \, ,
\end{equation}
with $h_\rho=h$ unique solution to $\rho= \partial_h \tf(\gb, h)$: note that, by  \eqref{eq:smoothing2}, by  the large $h$ remark at the end of \ref{P:smooth} and the strict convexity of  $\tf(\gb, \cdot)$ (see Theorem~\ref{th:unifconvex}),  $\partial_h \tf(\gb, \cdot)$ is a bijection from  $(h_c(\gb), \infty)$
to $(0,1)$. Fully exploiting Theorem~\ref{th:unifconvex}, i.e. using $\partial^2_h \tf(\gb, h)>0$ for $h>h_c(\gb)$,
by the Implicit Function Theorem we see that $\rho\mapsto h_\rho$ is $C^\infty$, so $\tg(\gb, \cdot)$ is $C^\infty$ too.

By differentiating once \eqref{eq:int54s} we obtain $\partial_\rho \tg(\gb, \rho)= -h_\rho$ which tends to $-h_c(\gb)$ for $\rho \searrow 0$. By differentiating once more we obtain $-\partial^2_\rho \tg(\gb, \rho)=1/\partial^2_h \tf(\gb, h_\rho)\in (0, \infty)$
and we have all the claimed estimates except \eqref{eq:Gsmooth}, that we consider next. 
Since $h_\rho \searrow h_c(\gb)$ when $\rho \searrow 0$,  for 
every $\gD_0>0$ we have
\begin{equation}
\label{eq:FGrestricted}
\tg(\gb, \rho)\,=\, - h_c(\gb) \rho+\inf_{h \in (h_c(\gb), h_c(\gb)+\gD)} \left( \tf(\gb, h)-\rho (h-h_c(\gb))  \right) \,,
\end{equation} 
provided that $\rho$ is smaller than  constant that depends on $\gD_0$. Therefore by \ref{P:smoothing1}
we obtain $\tg(\gb, \rho) \le - h_c(\gb)\rho -\rho^2/(4c_\gb)$ and \eqref{eq:Gsmooth} follows.

Finally, \eqref{eq:Gbetarho0-3} follows directly from \eqref{eq:irrelevant} of
\ref{P:irrelevant}: this analysis coincides with the $\gb=0$ analysis, developed in greater generality in Section~\ref{sec:3}.
\end{proof}

\begin{figure}[htbp]
\centering
\includegraphics[width=14 cm]{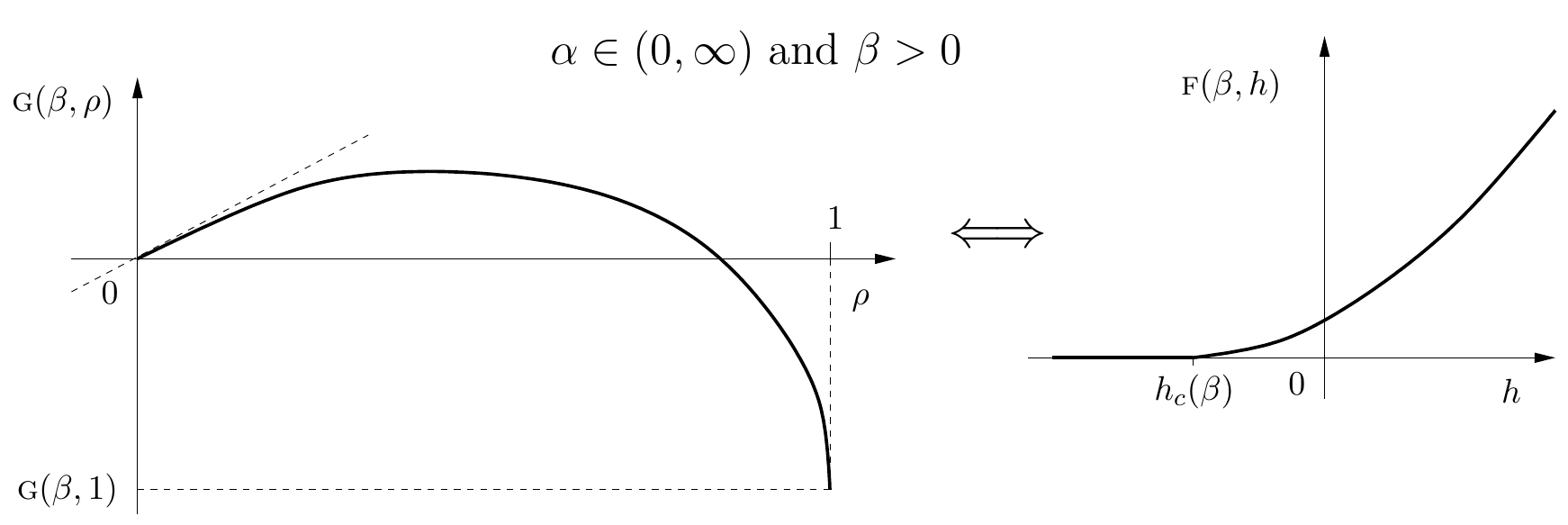}
\caption{\label{fig:alpha3} 
The figure illustrates the (Legendre transform) link between $\tg(\cdot)$ and $\tf(\cdot)$ for $\ga >0$ and $\gb>0$.
 $\tf(\gb, \cdot)$ has a non analyticity point at a critical value $h_c(\gb)<0$: to the left of this critical value the free energy is zero
 and to the right it is positive, $C^\infty$ and strictly convex. From this we can extract   $\tg(\gb, \cdot)$ is $C^\infty$ and strictly concave.
 The positive slope of $\tg(\gb, \cdot)$ at the origin, more precisely $\lim_{\rho \searrow 0}\partial _\rho \tg(\gb,\rho )= -h_c(\gb)$,  is a direct consequence  
  $h_c(\gb)<0$ (see Proposition~\ref{th:Gsmooth}). The fact that $\tg(\gb, \cdot)$ is strictly concave strongly hints to the similarity with the case of Figure~\ref{fig:alpha2}.
}
\end{figure}

\subsection{The generalized pinning model: free energy and transitions}
The free energy $\tf_H(\gb, h)$ is just given by \eqref{eq:FHG} via elementary considerations  given the properties and features 
of $\tg(\gb, \rho)$. These features are richer for $\gb=0$ (Figures~\ref{fig:alpha1} and \ref{fig:alpha2}) and they reduce to
Figure~\ref{fig:alpha3} for $\gb>0$. In particular
\smallskip

\begin{itemize}
\item $\tg(0, \rho)$ has a singularity at $\rho=\rho_c>0$ that directly reflects on a singularity of $\tf_H(0, h)$ at  $h_b= -H'(\rho_c)$, and
 corresponds to the big-jump transition: the proof is in Section~\ref{sec:path-beta0}, but the result can be readily understood because the variational formula suggests that the system will behave like a renewal constrained to a contact density $\rho=\rho_h$, where $\rho_h$ is the optimal density;
\item when $\gb>0$ instead this singularity disappears and, modulo the shift of the   critical point $h_c(0)$
to $h_c(\gb)$, that generates the positive slope at the origin,   Figure~\ref{fig:alpha3} is analogous to Figure~\ref{fig:alpha2}. Therefore the 
transition at $h_b$ disappears: the proof that the trajectories of the process have no big-jump transition is given in
Section~\ref{sec:pathsbeta>0}.
\end{itemize}
\smallskip

Here we provide the proof that Proposition~\ref{th:Gsmooth} yields, via Legendre transform, the properties of $\tf_H(\gb,h)$, for $\gb>0$, given in Theorem~\ref{th:beta>0}, see also Figure~\ref{fig:FH}. 

\medskip

\begin{proof}
[Proof of Theorem~\ref{th:beta>0}(1)]
The result is already known if $H(\cdot)$ is trivial, but the argument  we give applies in general. By \eqref{eq:FHG} and 
Proposition~\ref{th:Gsmooth} we have that
\begin{equation}
\label{eq:FHG5}
\tf_H(\gb, h)\,=\,   h \rho_h +H(\rho_h) + \tg(\gb, \rho_h)\, ,
\end{equation}
with $\rho=\rho_h$ unique solution of $h= -H'(\rho)-\partial_\rho\tg(\gb, \rho)$. Note that 
$\partial^2_\rho\tg(\gb, \rho)<0$ for $\rho \in (0,1)$ yields that $h \mapsto \rho_h$ is a $C^\infty$ bijection 
between $(h_c(\gb)-H'(0), \infty)$ and $(0,1)$. In particular
$\tf_H (h)> H(0)$ for $h> h_c(\gb)-H'(0)$ and,
 by continuity,  $\tf_H (h_c(\gb)-H'(0))= H(0)$. On the other hand, $\rho_h=0$ for $h < h_c(\gb)-H'(0)$.
 So $h_c^H(\gb)=h_c(\gb)-H'(0)$. The fact that 
 $\tf_H(\gb, \cdot) \in C^\infty$ on $\bbR \setminus \{h_c^H(\gb)\}$ is also a direct consequence of
\eqref{eq:FHG5} above $h_c^H(\gb)$, and of the triviality of the free energy below $h_c^H(\gb)$.

Let us turn to \eqref{eq:smoothingH}.
We claim that for every $\gb>0$ there exists a constant $c>0$ such that
 $\rho_h \le c (h-h_c^H(\gb))$ for $h$ sufficiently close to $h_c^H(\gb)$. This suffices to show \eqref{eq:smoothingH}
 because from the variational formula  \eqref{eq:FHG} 
 and $H(\rho)\le H(0)+ \rho H'(0)$, i.e. concavity, and $\tg(\gb, \rho) \le -h_c(\gb)\rho$ (Proposition~\ref{th:Gsmooth} )
  directly yield 
 \begin{equation}
\label{eq:varrhoh-1}
\tf_H(\gb, h)\, \le \, H(0)+ (h -h_c^H(\gb)) \rho_h\, ,
\end{equation} 
which is \eqref{eq:smoothingH} if we use the claim. To prove the claim we use the implicit characterization of $\rho_h$ for $h>h_c(\gb)$ that we write as 
\begin{equation}
h-h_c^H(\gb)\,=\, g_\gb(\rho_h) \ \text{ with } \ 
g_\gb(\rho)\,:=\,
-(H'(\rho)-H'(0))-\left(\partial_\rho\tg(\gb, \rho) +h_c(\gb)\right)\,.
\end{equation}
Note that $g_\gb(\cdot)$ is smooth and increasing and it satisfies  $g_\gb(\rho)\ge C \rho$
for a positive constant $C$: this is obvious if we assume \eqref{eq:assumefordeloc} (and in this case $C$ does not depend on $\gb$), 
but it is true in general  because  \eqref{eq:Gsmooth} implies $-\partial_\rho \tg(\gb, \rho)-h_c(\gb)\ge c \rho$ (this is simply because if $f(\cdot)$ is convex, $f'(\rho)\ge f(\rho)/ \rho$). Therefore $g_\gb^{-1}: [0, \infty) \longrightarrow [0,1)$ satisfies $g_\gb^{-1} (x) \le x/C$, so
\begin{equation}
 \rho_h \,=\, g_\gb^{-1}\left(h-h_c(\gb)\right) \, \le \, \frac{h-h_c(\gb)}C 
\, ,
\end{equation}
and the claim is proven. 
 
The proof of  \eqref{eq:sharpH} is  analogous to the one for $\gb=0$ (once again: the $\gb=0$ analysis is developed in detail in Section~\ref{sec:3}), because of 
the sharp estimate \eqref{eq:Gbetarho0-3}, which, by convexity, holds also if we formally differentiate both sides of the 
asymptotic equivalence.
\end{proof}

\begin{figure}[htbp]
\centering
\includegraphics[width=11 cm]{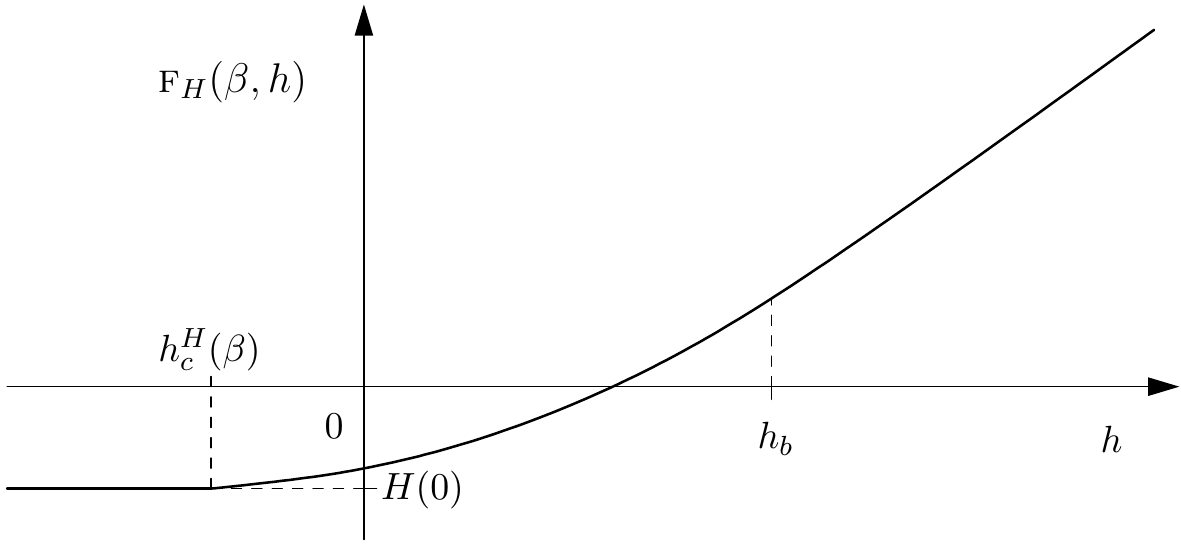}
\caption{\label{fig:FH} 
This is the graph of $\tf_H(\gb, h)$ for $H(\cdot)$ non trivial and $H'(0)<\infty$, hence $H(0)>-\infty$ too. 
We see the (de)localization transition 
at $h_c^H(\gb)=-H'(0)+ h_c(\gb)$ and we remark that the contact density is continuous at this transition. 
The transition at $h_b$ is instead present only if $\gb=0$ (and $\ga>1$), but it is difficult to appreciate it in this image because 
the contact density is also in this case continuous, or more regular, at the transition and the free energy is non trivial, i.e. non affine, both on the left and on the right of $h_b$.
}
\end{figure}

\section{Free energy in the  non disordered case: proof of Theorem~\ref{th:Fbeta0}}
\label{sec:3}
Recall that $\rho_c=1/\bE[\eta]\in [0, 1)$ and that $\rho_c=0$ if $\ga \in (0,1)$ and $\rho_c>0$ if $\ga>1$. In the next statement 
$c$ is a positive constant for which we have an explicit expression in terms of $\ga$, moments of $\eta$ and $C_K$: we are going to specify on what $c$ depends,
except for $C_K$ (see \eqref{eq:K}) that is omitted, because $c$ depends on $C_K$ in all the cases, either directly or via $\bE[\eta]$ 
and  $\bE[\eta^2]$.  
\medskip

\begin{proposition}[Basic properties of $\tg$]
\label{th:propG}
For every $\rho\in [0,1]$ the limit in \eqref{eq:G} with $\gb=0$ exists and 
\begin{equation}
\label{eq:lemG}
\tg(\rho)\,=\,\inf_{x\ge 0} \left( x + \rho \log \bE\left[ \exp(-x \eta)\right]\right)\,,
\end{equation}
from which the concavity of $\tg(\cdot)$ is evident and we have also the uniform estimate:
\begin{equation}
\label{eq:lemG-0}
\lim_{\gep \searrow 0 } \limsup_{N \to \infty} \suptwo{m\in \{1, \ldots, N:\}}{\vert m/N- \rho \vert \le \gep}
\left \vert
\tg(\rho)-
 \frac 1N \log \bP\left( \tau_m =N\right)\right\vert \,=\, 0\, .
\end{equation}
Moreover
\begin{enumerate}
\item if $\ga \in (0,1]$ (see Fig.~\ref{fig:alpha2}) then $\tg(\cdot)$ is analytic and negative in $(0,1)$ and, if we exclude $\ga=1$,
 for $\rho \searrow 0$
\begin{equation}
\label{eq:lemG-1}
\tg(\rho) \sim -c \rho^{1/(1-\ga)}  \ \text{ and } \ 
\tg'(\rho) \sim -\frac {c}{1-\ga}\rho^{\ga/(1-\ga)} \, ,
\end{equation}
and $c$ depends on $\ga$.
For $\ga=1$ instead $-\tg'(\rho)= \exp(-(1+o(1))/(C_K \rho))$ and $-\tg(\rho)=o(-\tg'(\rho))$. 
\item if $\ga >1$ (see Fig.~\ref{fig:alpha1}) then $\tg(\rho)=0$ for $\rho \in [0,\rho_c]$ and $\tg(\cdot)$ is analytic and negative 
on $(\rho_c,1)$. Moreover with $\gd:= \rho -\rho_c$ and $\kappa:=\max(\ga/(\ga-1), 2)$  in the limit $\gd \searrow 0$
we have for $\ga\neq 2$
\begin{equation}
\label{eq:lemG-2}
\tg(\rho_c+ \gd) \sim -c \, \gd^{\kappa}
 \ \text{ and } \ \tg'(\rho_c+\gd) \sim - \kappa c \,  \gd^{\kappa-1}\,, 
 \end{equation}
 where $c$ depends on $\ga$ and $\bE[\eta]$  when $\ga\in (1,2)$ and it depends on $ \bE[\eta]$ and $\bE[\eta^2]$ if $\ga>2$.
 When $\ga=2$ we have instead $\tg(\rho) \sim c \gd^2/ \log(1/\gd)$ and $\tg'(\rho) \sim 2c \gd/ \log(1/\gd)$ with  $c$ that depends on $\bE[\eta]$.
\item $\lim_{\rho \nearrow 1}\tg(\rho)= \tg(1)= \log K(1)  \in (-\infty,0)$ and $\lim_{\rho \nearrow 1}\tg'(\rho)=-\infty$.
\end{enumerate}
\end{proposition}

\medskip

\begin{proof}
Existence of $\tg(\rho)$ and  \eqref{eq:lemG} can be established at the same time by standard arguments: for the upper bound it suffices to apply the 
Markov inequality to $\bP\left( \tau_m \le N\right)=\bP\left( \exp(-x \tau_m) \ge \exp(-x N)\right)$, for $x>0$; for the lower bound 
the standard exponential tilt argument covers the case $\rho>\rho_c$, while for $\rho < \rho_c$ the lower bound is easily achieved 
by selecting trajectories that make a suitable big-jump (no exponential cost), so that in the rest of the system the contact density is $\rho_c$ and this matches with the typical behavior of the renewal (again, no exponential cost): if 
$\rho=\rho_c$ the argument is the same, but the big-jump is empty. Details of the proof can be found 
in \cite[\S~4.3]{cf:thesis-BH};  in a more general context these estimates can be found for example  in \cite{cf:BB00,cf:BM2}.

Now we set $g_\rho(x):=  x + \rho \log \bE[ \exp(-x \eta)]$ and remark that for every $\rho\in (0,1)$ the function
$g_\rho(\cdot)$ is strictly convex and $\lim_{x\to \infty} g_\rho(x)=\infty$ because in this limit
$\bE[ \exp(-x \eta)] \sim K(1)\exp(-x)$. Therefore the infimum in the right-hand side of \eqref{eq:lemG}
is reached at a unique point $x_\rho\ge 0$. Since for $x>0$
\begin{equation}
\label{eq:sec2tilt}
g'_\rho(x)\, =\, 1- \rho \frac{\bE[ \eta \exp(-x \eta)]}{\bE[ \exp(-x \eta)]}\, =:\, 1- \rho \bE_x[\eta]\, ,
\end{equation} 
we readily see that if $\rho \le \rho_c$ then $g'_\rho(x)>0$ for every $x>0$ so $x_\rho= 0$ and 
$\tg(\rho)=g_\rho(0)=0$. If $\rho > \rho_c$ instead  $x_\rho> 0$, 
because in this case 
 $g'_\rho(x)=0$, that is $\rho \bE_x[\eta]=1$ can be solved with $x\in (0,1)$. Since $x\mapsto \log \bE[ \exp(-x \eta)]$ is real analytic on $(0, \infty)$ 
from the analytic implicit function theorem one readily obtains that $\rho \mapsto x_\rho$ is analytic on $(\rho_c, 1)$
and this property passes directly to $\tg(\cdot)$, because $\tg(\rho)=g_\rho(x_\rho)$.  

The rest of the proof is concerned with the asymptotic behaviors for  $\rho \searrow \rho_c$ and $\rho \nearrow 1$. Key formulas for this are
$\bE_{x_\rho} [\eta]=1/\rho$ (cf. \eqref{eq:sec2tilt}) and 
\begin{equation}
\label{eq:sec2basic}
\tg(\rho)\,=\, x_\rho + \rho \log \bE\left[ \exp(-x_\rho \eta)\right] \ \text{ and } \
 \tg'(\rho)\,=\, \log  \bE\left[ \exp(-x_\rho \eta)\right]\, .
\end{equation}

For the case $\rho \nearrow 1$ we 
observe that, for  $x  \nearrow \infty$, 
both $\bE[\exp(-x \eta)]$ and $\bE[\eta \exp(-x \eta)]$ are equal to $K(1) \exp(-x)+O( \exp(-2x))$.
 Therefore $\bbE_x[\eta]= 1+O( \exp(-x))$  
 and this implies 
 $x_\rho \nearrow \infty$ as   
 $\rho\nearrow 1$ with $  1-\rho= O(\exp(-x_\rho))$.
 Therefore 
 \begin{equation}
 \tg(\rho)\, =\, x_\rho + \rho \log \left( K(1)e^{-x_\rho} +O\left(e^{-2x_\rho} \right) \right)\, =\, 
 x_\rho
 (1-\rho) + \rho \log K(1)   + O\left(e^{-x_\rho} \right) 
 \,,
 \end{equation}
 so $\lim_{\rho \nearrow 1} \tg(\rho)= \log K(1)$. By using the second expression in \eqref{eq:sec2basic}, we get 
 $\tg'(\rho)\sim -x_\rho$, in particular
 $\tg'(1)=-\infty$.
 
We are left with $\rho \searrow \rho_c$ that we separate into $\rho_c=0$ and $\rho_c>0$. Let us first remark that in both cases
$\lim_{\rho \searrow \rho_c} \tg(\rho)= \tg(\rho_c)=0$. This is obvious by concavity when $\rho_c>0$.
If $\rho_c=0$ it suffices to use the first expression in \eqref{eq:sec2basic} and the fact that $\bE_{x_\rho}[\eta]=1/\rho \nearrow \infty$ when $\rho \searrow 0$, so $x_\rho\to 0$  in this limit.
This remark simplifies the analysis because the  asymptotic analysis of $\tg(\cdot)$ near $\rho_c$ follows from by integrating the corresponding estimate on  $\tg'(\cdot)$.

\smallskip

For $\rho_c=0$, i.e. $\ga\in (0,1]$, by Riemann sum approximation 
we readily find
that $\bE_{x} [\eta]\sim \bE[\eta \exp(-x \eta)]$ for $x \searrow 0$ and 
\begin{equation}
\label{eq:frjh3}
\bE[\eta \exp(-x \eta)] \stackrel{x \searrow 0} \sim C_K \times \begin{cases}
\left(\int_0^\infty y^{-\ga} e^{-y} \dd y\right) x^{\ga-1}\, & \textrm{ if } \ga \in (0,1)\, ,
\\
\log(1/x)  \textrm{ if } \ga =1\, .
\end{cases}
\end{equation} 
Of course $\int_0^\infty y^{-\ga} e^{-y} \dd y=\gG (1- \ga)$, but we will not keep track of the precise value of the constants and we content ourselves with remarking that we have obtained for $\ga \in (0,1)$ that $x_\rho \sim c_\ga \rho^{1/(1-\ga)}$. 
Now we can insert this result into the second identity in \eqref{eq:sec2basic} that in this limit becomes $\tg'(\rho)
\sim -\bE[1- \exp(-x_\rho\eta)]$: the sharp asymptotic behavior of the right-hand side is again a matter of Riemann sum approximation for $x_\rho$ that tends to zero. So for 
for $\ga \in (0,1)$ we obtain $\tg'(\rho) \sim -C_K(\Gamma(1-\ga)/\ga) x_\rho^\ga \sim -c \rho^{\ga/(1-\ga)}$. 
For $\ga=1$, going back to \eqref{eq:frjh3} we see that 
$x_\rho=\exp(-(1+o(1))/(C_K \rho))$, so by using $\bE[1- \exp(-x\eta)]\sim C_K x \log (1/x)$ from which, using \eqref{eq:sec2basic},
we obtain $\tg'(\rho) \sim -C_K x_\rho \log (1/x_\rho)$, which implies $\tg'(\rho)=\exp(-(1+o(1))/(C_K \rho))$ and, by convexity  of $-\tg(\cdot)$, 
we see that $0 \le \tg(\rho)/\tg'(\rho)\le \rho$.

\smallskip

For $\rho_c>0$, i.e. $\ga>1$, the analysis is different according to whether $\bE[\eta^2]< \infty$ or not:
\begin{enumerate}
\item if $\bE[\eta^2]< \infty$, that is $\ga>2$, and $x \searrow 0$ we have
\begin{equation}
\bE_x[ \eta]\, =\, \frac{\bE[\eta \exp(-x\eta)]} {\bE[ \exp(-x\eta)]} \,=\, \frac{\bE[\eta]-\bE[\eta^2]x (1+o(1))}{1-\bE[\eta] x(1+o(1)) }\, =\,
\bE[\eta]- \mathrm{var} (\eta) x 
+o(x)\, .
\end{equation}
Therefore $x_{\rho_c+\gd} \sim (\bE[\eta]^2/\mathrm{var} (\eta))\gd$, so $\tg'(\rho)\sim  -\bE[\eta] x_\rho$ directly yields the result
for $\tg'(\rho_c+\gd)$. 
\item if $\bE[\eta^2]= \infty$ we consider separately $\ga \in (1,2)$ and $\ga=2$. In the first case we use 
\begin{equation}
\bE_x[\eta]\, =\, \frac{\bE[\eta\exp(-\eta x)]}{1+O(x)}\,=\, \bE[\eta] -C_K \left( \int_0^\infty \frac{1-e^{-y}}{y^\ga} \dd y \right) x^{\ga-1}
+o\left(x^{\ga-1}\right)\, ,
\end{equation}
so $x_{\rho_c + \gd} \sim c_\ga \gd^{1/(\ga-1)}$, $\tg'(\rho_c + \gd)\sim -\bE[\eta] x_{\rho_c + \gd}
\sim -\bE[\eta] c_\ga \gd^{1/(\ga-1)}$. 
For $\ga=2$ we have $\bE_x[\eta]=\bE[\eta]- C_K(1+o(1)) x \log (1/x)$ that entails 
$\gd \sim \rho_c^2 C_K x_{\rho_c+\gd} \log x_{\rho_c+\gd} $ so 
$x_{\rho_c+\gd} \sim \gd /(\rho_c^2 C_K \log(1/\gd))$ and with $\tg'(\rho) \sim -\bE[\eta] x_\rho$ we conclude. 
\end{enumerate}
\end{proof}

\medskip

\noindent
\emph{Proof of Theorem~\ref{th:Fbeta0}.}
Of course we are going to use intensively
\begin{equation}
\label{eq:FHG0}
\tf_H( h)\,=\, \sup_{\rho \in [0,1]} \left( h \rho +H(\rho) + \tg( \rho)\right)\, , \ \
\tf'_H( h)\,=\, \rho_h
\ \text{ and } \ \ 
H'(\rho_h)+  \tg'( \rho_h)\, =\, -h\,,
\end{equation} 
where  the second equation identifies the unique optimizer  $\rho_h$, as long as  $\rho_h>0$. 
So the second and third identity are written for  $\rho_h>0$. Note also that
\eqref{eq:lemG-1} implies that $\tg(\cdot)$ is $C^1$ also at $\rho_c$.

We start with the case $h_c^H=-H'(0)>-\infty$: one readily sees that $\rho_h=0$ for $h \le h_c^H$ and, by the first identity in 
\eqref{eq:FHG0},
$\tf_H( h)=H(0)$ for these values of $h$. On the other hand, for $h>h_c^H$,
the second identity in \eqref{eq:FHG0} can be written as 
$-(H'(\rho_h)-H'(0))-   \tg'( \rho_h)=(h-h_c^H)$  and we see that it
has a strictly positive solution $\rho_h(=\tf_H'(h))$ because the left-hand side is an increasing function of $\rho_h$, and this directly yields
 $\tf_H(h) > H(0)$. Therefore $h_c^H$ is a critical (i.e., non analyticity) point. 
 
 We now recall that we assume \eqref{eq:assumefordeloc}. Here is the $h \searrow h_c^H$ analysis:
 
 \smallskip
 \begin{itemize}
 \item if $\rho_c>0$ then $h-h_c^H=-H'(\rho_h)+H'(0) \sim c_H \rho_h$, so $\tf_H'(\rho)=\rho_h\sim (h-h_c^H)/c_H$
 and $\tf_H(\rho)-H(0)\sim (h-h_c^H)^2/(2c_H)$;
 \item
 if $\rho_c=0$ then $h-h_c^H=-H'(\rho_h)+H'(0)-\tg'(\rho_h)$ and, by  Proposition~\ref{th:propG}(1), we have (with $c_\ga=c/(1-\ga)$)
 \begin{equation}
 h-h_c^H\, \sim
 \begin{cases}
 c_\ga \rho_h^{\ga/(1-\ga)} & \text{ if } \ga\in (0, 1/2)\, ,
 \\
 (c_{1/2}+c_H) \rho_h & \text{ if } \ga= 1/2\, ,
 \\
 c_H \rho_h & \text{ if } \ga \in (1/2,1]\, ,
 \end{cases}
 \end{equation}
 and, like above, from $\tf_H'(\rho)=\rho_h$ we extract the claimed asymptotic behaviors.
 \end{itemize}

\medskip

We turn now to $h_b$, of course when  $h_b>h_c^H$ (so $\ga >1$ and $H(\cdot)$ is  non  trivial) otherwise we are in the case we just considered.
 The origin of the $h_b$ singularity is simply the fact that 
 $\tf_H(h)$ is determined by different variational problems according to whether $h< h_b$ and $h>h_b$.
 In fact 
  $h\le h_b$ 
means $\rho_h\le \rho_c$, i.e. $\tg(\rho)=0$, and 
the variational problem in this case reduces to $\tf_H( h)= \tf_H^{\mathrm{reg}}( h)=\sup_{\rho \in [0,1]} \left( h \rho +H(\rho) \right)$.
For $h>h_b$ instead 
$\rho_h> \rho_c$, i.e. $\tg(\rho)>0$, and one has to use the full expression for $\tf_H(\cdot)$.

Let us analyse the singularity at $h_b$. 
We start by remarking that we are just need to do a  local analysis at $h_b=-H'(\rho_c)$: by introducing $J(y):= H(\rho_c+y) -   H(\rho_c) -H'(\rho_c)y$
and $G(y)=\tg(\rho_c +y)$, we can work with
\begin{equation}
\label{eq:JG}
F(x)\, :=\, \sup_y \left( xy +J(y) +G(y)\right)\ \ \text{ and } \  \ F^{\text{reg}}(x)\, :=\, \sup_y \left( xy +J(y) \right)\, .
\end{equation}  
The maximizer   $y_x$ for   $F(\cdot)$
 is the (unique) solution of $U'(y_x)=-x$ with $U=J+G$. We remark from the start that $F(x) \le F^{\text{reg}}(x)$ and the inequality is strict for $x>0$. 
 
Consider first the case of $\kappa$ which  is not an integer and set $\mathtt{k}= \lfloor \kappa \rfloor$, so $\mathtt{k} =2, 3, \ldots$:
this means that we are considering $\ga \in (1,2)$.
For $y \searrow 0$ we have
\begin{equation}
\label{eq:Taylorsv1}
U(y)\, =\, -a_2 y^2-a_3 y^3- \ldots - a_{\mathtt{k}} y^\mathtt{k} - b_\kappa y^\kappa  +o( y^\kappa)\, ,
\end{equation}
where $b_\kappa>0$ is the constant $c$ in  \eqref{eq:lemG-2}. 
Moreover, still by   \eqref{eq:lemG-2} , we have 
\begin{equation}
\label{eq:Taylorsv2}
U'(y)\, =\, -2a_2 y-3a_3 y^2- \ldots - \mathtt{k}a_{\mathtt{k} } y^{\mathtt{k}-1} - \kappa b_\kappa y^{\kappa-1} 
+ o\left( y^{\kappa-1} \right)\, .
\end{equation}
Note that $a_2= \vert J''(0)\vert>0$, but the other $a_\cdot$ coefficients are just real numbers.
From \eqref{eq:Taylorsv2} we extract that as $x \searrow 0$
\begin{equation}
F'(x)\, =\, y_x\, =\, c_1 x +c_2x^2+ \ldots + c_{\mathtt{k} -1} x^{\mathtt{k}-1} - c_\kappa x^{\kappa-1}
+o\left( x^{\kappa-1}\right)\, , 
\end{equation}
where  $c_1=1/(2a_2)$ and
$c_\kappa=  \kappa b_{\kappa}/{(2a_2)^\kappa}$.
so 
\begin{equation}
\begin{split}
F(x)\, &=\,  \frac 12 c_1 x^2 +\frac 13c_2x^3+ \ldots + \frac 1{\mathtt{k}} c_{\mathtt{k} -1} x^{\mathtt{k}} - \frac 1{\kappa} c_\kappa x^{\kappa-1}
+ o\left( x^{\kappa-1}\right)
\\
&=: \,P_\mathtt{k} (x)- \frac 1{\kappa} c_\kappa x^{\kappa-1}
+ o\left( x^{\kappa-1}\right)\, ,
\end{split} 
\end{equation}
  where the last line defines $P_\mathtt{k} (x)$, a polynomial of degree $\mathtt{k}$.

 Since the analysis we have developed  can be repeated for $F^{\mathrm{reg}}(x)$, that is with
 $U=J$,   in an essentially identical (in fact, simpler) way, and considering $x \to 0$ (not simply $x \searrow 0$)
 we readily see that $F^{\mathrm{reg}}(x)= P_\mathtt{k} (x) +O(x^{\mathtt{k}+1} )$. This completes the case of $\kappa$ non integer.

\medskip

Let us consider now the cases $\kappa=\mathtt{k}=2,3, \ldots$.

When $\kappa=2$, that is $\ga=2$ (and  $\bE[\eta^2] = \infty$),
 for $x \searrow 0$ we have 
$U'(y)= -2a_2 y- 2b_2 y/ \log (1/y)+$h.o., of course h.o. means higher order. Therefore $y_x=x/(2a_2)- (b_2/(2 a_2^2))x/\log(1/x) +$h.o., which yields 
 $F(x)= x^2/(4a_2)+ (b_2/(4 a_2^2))x^2/\log(1/x) +$h.o..
 For $x\nearrow 0$ it suffices to repeat the  same analysis, but this time there is no logarithmic terms and 
 $(b_2/(4 a_2^2))x^2/\log(1/x)$  becomes simply  $O(x^3)$.

For $\kappa=\mathtt{k}=3,4, \ldots$ the analysis changes slightly because  the last two  explicit terms in the right-hand side
of \eqref{eq:Taylorsv1} and  \eqref{eq:Taylorsv2} have the same  behavior.
Therefore 
the coefficient appearing in front
of the term $y^\mathtt{k}$ in the development for $U(y)$ is $a_\mathtt{k}+b_\mathtt{k}$, respectively   $a_\mathtt{k}$,  when  $y\searrow 0$, respectively
 $y \nearrow 0$. This mismatch directly reflects on a mismatch in the $x^\mathtt{k}$ term of the two developments for $F(x)$.
 \qed

\section{Path properties in the  non disordered case: proof of Theorem~\ref{th:paths0}}
\label{sec:path-beta0}

The basic step is observing that the probability $\bP_{N, h}^\Psi $ introduced in \eqref{eq:mod-without-dis} is a superposition of 
probabilities in which the number of contacts is fixed: 
\begin{equation}
 \label{eq:dec-contact}
 \bP_{N, h}^\Psi (\cdot)\, :=\, 
 \frac{ \sum_{m=1}^N\exp(hm) \Psi(m, N) \bP(\tau_m=N)\bQ_{N, m} (\cdot)}
 {\sum_{m=1}^N\exp(hm) \Psi(m, N) \bP(\tau_m=N)} 
 \, ,
 \end{equation} 
 where $\bQ_{N, m}$ is the law of $\tau \cap [0,N]$ conditioned to $\tau_m=N$. So $\bQ_{N, m}$ is the law of the renewal conditioned to visiting $N$ in precisely $m$ steps and can of course be viewed as a non disordered pinning model conditioned to having $m$ contacts (one of which is at $N$). But this process is very relevant well beyond pinning models and in fact it has been studied in depth: we collect here the results we will use (that are only a minimal part of what is available in the literature).

\medskip

Recall the notations introduced for Theorem~\ref{th:paths0}: 
\medskip

\begin{proposition}
\label{th:muNm} 
In two parts, the first applies only to $\ga>1$, the second one is general:
\begin{enumerate}
\item 
for every $\rho \in [0, \rho_c)$ and every $\gep\in (0, (\rho_c-\rho)/2)$ we have
\begin{equation}
\label{eq:muNm1}
\lim_{N \to \infty} \inf_{m:\, \vert m/N -\rho \vert \le \gep }
\bQ_{N, m} \left(
 \left \vert \frac {\eta_{1,N}}N-\left(1-\frac{\rho}{\rho_c}\right)
\right \vert +   \frac {\eta_{2,N}}N  \le \frac{2\gep}{\rho_c} 
\right)\,=\, 1\,.
\end{equation}
\item for every $\rho\in (\rho_c, 1)$
and every $\gep \in (0, \min(\rho-\rho_c,1-\rho)/2)$ 
\begin{equation}
\label{eq:muNm2}
\lim_{N \to \infty} \inf_{m:\, \vert (m/N) -\rho \vert \le \gep }
\bQ_{N, m} \left(
 \left \vert  \frac {\eta_{1,N}}N \right \vert \le \gep
\right)\,=\, 1\,.
\end{equation}
\end{enumerate}
\end{proposition}

\medskip

\begin{proof}
The first part is a result in the big-jump domain and one can directly apply the (much sharper and much more general) result
in \cite[Th.~1]{cf:AL} (see also \cite{cf:DDS08}) that implies that  for $\rho < \rho_c$ and every $\tilde \gep>0$
\begin{equation}
\label{eq:muNm1-1}
\limtwo{N \to \infty:}{m/N \sim \rho} 
 \bQ_{N, m} \left(
 \left \vert \frac {\eta_{1,N}}N-\left(1-\frac{\rho}{\rho_c}\right)
\right \vert +   \frac {\eta_{2,N}}N  \,>\,\tilde \gep
\right)\,=\, 0
\,.
\end{equation}
Therefore if we set
\begin{equation}
\label{eq:muNm1-2}
p_{N, m}\, :=\, \bQ_{N, m} \left(
 \left \vert \frac {\eta_{1,N}}N-\left(1-\frac{\rho}{\rho_c}\right)
\right \vert +   \frac {\eta_{2,N}}N  > \frac{2\gep}{\rho_c} 
\right)
\,,
\end{equation}
then $\lim_{N\to \infty} p_{N, m_N}=0$ if $m_N/N \to  \rho < \rho_c$. This implies that, with $\rho$ and $\gep$ like in the statement 
$\lim_{N\to \infty} \sup_{m:\vert m/N -\rho \vert \le \gep} p_{N, m_N}=0$ because otherwise
there exists $p>0$ and a subsequence $(N_j)_{j\in \bbN}$ such that $\lim_j m_j/N_j \in  [\rho-\gep, \rho+ \gep]$ and 
$p_{N_j, m_j}\ge p$ for every $j$, in contradiction with $\lim_N p_{N, m_N}=0$.

The second part is in the Large Deviation regime and we can perform the standard \emph{tilting procedure} in a direct way 
because of the constraint that there are exactly $m$ contacts and the last one is in $N$. Explicitly:
\begin{equation}
\bQ_{N, m} \left(A\right)\, =\, {\bP\left(\tau^{(q)} \in A, \, \tau^{(q)}_m=N\right)}\Big/{\bP\left( \tau^{(q)}_m=N\right)}\, ,
\end{equation}
with $\tau^{(q)}$ the renewal process with inter-arrival probability distribution given by $K_q(n) \propto K(n)\exp(-qn)$,
and $q=q(m/N)$ is the unique solution of $\mu_q:=\sum_n n K_q(n)= N/m$.
We can now apply the Local Central Limit Theorem for triangular arrays (see e.g. \cite[Th.~1.2]{cf:DmcD}) that gives 
\begin{equation}
\label{eq:LCLT}
\lim_{m \to \infty}
\sup_{M\in \bbN} \left \vert \sqrt{ \gs^2_q m}\,  \bP\left( \tau^{(q)}_m=M\right) - f_\cN \left( \frac{M-m \mu_q}{ \sqrt{ \gs^2_q m}}\right)\right\vert  \, =\, 0\, ,
\end{equation}
with $\mu_q$ and $\gs^2_q$ respectively sum and variance of $ \tau^{(q)}_1$, and $f_\cN(\cdot)$ is the density of a standard Gaussian variable. From \eqref{eq:LCLT} we readily extract that 
with $\rho$ and $\gep$ as required in the statement there exists $c>0$ and $m_0>0$ such that, uniformly in $\vert m/N -\rho \vert  \le \gep$ we have that
$ \bP\left( \tau^{(q)}_m=N\right)\ge 1/(c\sqrt{m})$ for every $m \ge m_0$. 
Therefore 
\begin{equation}
\bQ_{N, m} \left(
 \left \vert  \frac {\eta_{1,N}}N \right \vert > \gep\right) \, \le \, c \sqrt{m} \, \bP \left( \sup_{j=1, \ldots, m} \tau^{(q)}_j-  \tau^{(q)}_{j-1}> \gep N\right)\,\le \, cm^{3/2} \bP \left( \tau^{(q)}_1> \gep N \right)\, ,
\end{equation}
and the right-most term vanishes, with an exponential rate, when $m\to \infty$. 
\end{proof}

Now we observe that, in analogy with \eqref{eq:dec-contact}
we have 
\begin{equation}
 \label{eq:dec-contact-Z}
 Z^\Psi_{N, h}\, =\, 
 \sum_{m=1}^N\exp(hm) \Psi(m, N) \bP(\tau_m=N)
 \, ,
 \end{equation} 
 
 Recall that we use  $\rho_h= \tf'_H(h)$ for $h \neq h_c^H$, cf. \eqref{eq:paths0}, 
 so $\rho_h$ is the unique point that maximizes $\rho \mapsto \rho h + H(\rho)+\tg(\rho)$.
 
 The central estimate for the proof of Theorem~\ref{th:paths0} is:
 
 \medskip
 
 \begin{lemma}
 \label{th:Zeps}
 For every $\gep>0$ and every $h \neq h_c^H$
 \begin{equation}
 Z^\Psi_{N, h} \stackrel{N \to \infty} \sim Z^{\Psi,\gep}_{N, h}\, :=\, 
  \sumtwo{m\in \{1, \ldots,N\}:}{\vert m/N- \rho_h  \vert \le \gep }
  \exp(hm) \Psi(m, N) \bP(\tau_m=N)
 \, .
 \end{equation} 
 \end{lemma}
 
 \medskip
 
 \begin{proof}
 We have 
 \begin{equation}
 \label{eq:fordisalso}
 \begin{split}
0 \, \le \,  Z^\Psi_{N, h} - Z^{\Psi,\gep}_{N, h}\, &\le \, 
  c \sumtwo{m\in \{1, \ldots,N\}:}{\vert m/N- \rho_h  \vert > \gep }
  \exp\left( (h+b) m +N H(m/N)\right) \bP(\tau_m=N)
  \\
 & \le \, c^2   \sumtwo{m\in \{1, \ldots,N\}:}{\vert m/N- \rho_h  \vert > \gep }
  \exp\left( (h+2b) m +N H(m/N)+ N\tg\left(m/N\right)\right)
  \\
 & \le \, c^2 N  
  \exp\left(2b N + N \sup_{\rho:\, \vert \rho -\rho_h\vert >\gep}\left(  \rho h  + H(\rho)+ \tg(\rho)\right)\right)
 \, ,
 \end{split}
 \end{equation} 
 where in the first step
 we have used the hypothesis
 \eqref{eq:defPsi3}, so $b>0$ can be chosen arbitrarily small and $c=c(b)$ just needs to be chosen sufficiently large.
 In the second step instead we used \eqref{eq:lemG-0} of Proposition~\ref{th:propG}. 
 Now it suffices to remark that 
 \begin{equation}
 \sup_{\rho:\, \vert \rho -\rho_h\vert >\gep}\left(  \rho h  + H(\rho)+ \tg(\rho)\right) \,=\,  \tf_H (h)-q_\gep\, ,
 \end{equation}
 with $q_\gep>0$ (here we use  $h\neq h_c^H$, but we stress that this is needed only if $H(\cdot)$ is trivial) and therefore,
 by choosing $b= q_\gep/5$ for $N$ sufficiently large, we have
 \begin{equation}
0 \, \le \,  Z^\Psi_{N, h} - Z^{\Psi,\gep}_{N, h}\, \le \, \exp\left(  \tf_H (h)-q_\gep/2\right)\, .
\end{equation}
Since $\log Z^\Psi_{N, h}\sim N  \tf_H (h)$ we are done.
 \end{proof}
 
 \medskip
 
 \noindent
 \emph{Proof of Theorem~\ref{th:paths0}.}
  \eqref{eq:paths0} 
  follows because
 \begin{equation}
 \bP_{N, h}^\Psi \left(A^0_N \right) 
 \,=\, \frac{Z^{\Psi,\gep}_{N, h}} {Z^{\Psi}_{N, h}} \stackrel{N \to \infty} \longrightarrow 1
 \ \text{ with } \ A^0_N\, :=\,  \left \{\tau :\, \left \vert
 \frac {\vert \tau \cap (0, N]\vert}N - \rho_h \right\vert \le \gep\right\}\, , 
 \end{equation}
 by Lemma~\ref{th:Zeps}. For $\rho_c>0$ ($\ga>1$ and we assume $h\neq h_b$) we
 consider  the event
 \begin{equation}
 A_N^1\, :=\, \left\{\tau:\, 
 \left \vert \frac{\eta_{1,N}}N - \left( 1- \frac{\rho_h}{\rho_c}\right)_+ \right \vert +
 \frac{\eta_{2,N}}N 
 \le \frac{2\gep}{\rho_c} \right\}\, .
 \end{equation}
  By Lemma~\ref{th:Zeps} we have
  \begin{equation}
  \label{eq:ke7g}
 \bP_{N, h}^\Psi \left(A^0_N \cap A^1_N \right) \stackrel{N\to \infty} \sim 
 \frac 1{Z^{\Psi, \gep}_{N, h}}
 \sum_{m:\, \vert m/N - \rho_h \vert \le  \gep } 
 \exp(hm) \Psi(m,N) \bP(\tau_m=N) \bQ_{N, m} \left(  A^1_N \right)\, .
 \end{equation}
 Recall now that we assume $h \neq h_b$, so $\rho_h\neq \rho_c$.
 By Proposition~\ref{th:muNm} we have that for $\gep$ sufficiently small
 $\bQ_{N, m} \left(  A^1_N \right)$ tends to one as $N \to \infty$, with the constraint we have on $m$,
 and this readily entails that numerator and denominator in the right-hand side of \eqref{eq:ke7g} are asymptotically equivalent,
 so \eqref{eq:paths0-1} is established. In the case $\rho_c=0$ we change the event $A_N^1$,  but the argument is the same.
 \qed

\section{Path properties in the  disordered case: proof of Theorem~\ref{th:beta>0}(2)}
\label{sec:pathsbeta>0}

We start with an estimate on the disordered pinning model that is in the spirit of the sharper, but also, to a certain extent,  different (see Remark~\ref{rem:verylong}), result in \cite[Theorem~2.5]{cf:GTalea}.  

\medskip

\begin{lemma}
\label{th:verylong}
Consider  the $\Psi \equiv 1$ model for $h$ such that $\tf(\gb,h)>0$. Then for every $\gamma\in (0,\tf(\gb,h)]$ there exists an a.s. finite random variable $N_0(\gamma, \go)$  such that for $N \ge N_0(\gamma, \go)$
\begin{equation}
\label{eq:verylong}
\bP_{N , \go, \gb , h} \left( \eta_{1,N}> \gamma N\right) \, \le \, \exp \left( - \gamma N \tf(\gb, h)/2 \right)\,.
\end{equation}
\end{lemma}
\medskip

\begin{proof}
Set $\bP_{N , \go} =\bP_{N , \go, \gb , h}$.
We also choose $\gamma \in (0, \tf(\gb, h)/2]$.
The key estimate is 
\begin{multline}
\bP_{N, \go}\left(\eta_{1,N}> \gamma N\right)\, =\, \sumtwo{ n_1,n_2\in \{1, \ldots, N\}:}{n_2-n_1 > \gamma N}
\frac{Z_{n_1, \go}  K(n_2-n_1) e^{\gb \go_{n_2}}Z_{n_2, \theta^{n_2}\go}}{Z_{N, \go} }
\\
\le\, \sumtwo{ n_1,n_2\in \{1, \ldots, N\}:}{n_2-n_1 > \gamma N}
\frac{Z_{n_1, \go}  K(n_2-n_1) e^{\gb \go_{n_2}}Z_{n_2, \theta^{n_2}\go}}
{Z_{n_2, \go}  
Z_{n_2, \theta^{n_2}\go}}
\,\le\, \sumtwo{ n_1,n_2\in \{1, \ldots, N\}:}{n_2-n_1 > \gamma N}
\frac{ Z_{n_1, \go} }  
{Z_{n_2, \go}  } e^{\gb \go_{n_2}}
.
\end{multline}
Now we observe that  $Z_{n_1, \go}\le C(\go) \exp(n_1 \tf(\gb,h)+\gamma/6)$ for every $n_1\in \bbN$
and $ Z_{n_2, \go} \ge \exp(n_2 \tf(\gb,h)-\gamma/6)$ for every $n_2 \ge \gamma N$ and 
for $N$ larger than a  random threshold, possibly  dependent  also on $\gamma$.
Therefore 
\begin{equation}
\label{eq:forltgg6}
\bP_{N, \go}\left(\eta_{1,N}> \gamma N\right)\, \le \, C(\go) N^2 \exp( - (\tf(\gb,h) -\gamma/3) \gamma N)\, \le \, 
\exp( - \gamma N\tf(\gb,h)/2)\, ,
\end{equation}
always for $N$ larger than a random threshold.
\end{proof}

\medskip

\begin{rem}
\label{rem:verylong}
Lemma~\ref{th:verylong}
is a rough version of the sharp statement (in probability) for the size of the largest excursion 
in the localized phase 
\cite[Th.~2.5]{cf:GTalea}. In \cite[Th.~2.5]{cf:GTalea} the size of the largest excursion in the localized phase 
is shown to be equal, in $\bbP(\dd \go)$ probability and to leading order, to $c(\gb,h) \log N$, for a suitable choice of $c(\gb,h)>0$. 
Lemma~\ref{th:verylong} however, with respect to  \cite[Th.~2.5]{cf:GTalea}, gives a
 $\bbP(\dd \go)$-almost sure estimate and, above all, an exponential decay rate in the quenched probability 
proportional to $N$.
\end{rem}

\medskip

Let us consider now the constrained case. 
Recalling \eqref{eq:G} we set
\begin{equation}
\label{eq:cZEG}
\cZ_{N,m, \go, \gb}\, :=\,  
  \bE \left[ \exp\left( \gb  \sum_{j=1}^N  \go_j \gd_j \right) \ind_{ \tau_m=N}\right]\,,
\end{equation}

\medskip

\begin{proposition}
\label{th:nobigjump}
Choose any $\rho \in (0,1)$. Then for every $\gamma>0$ there exists $\gep_0$ such that for $\gep \in (0, \gep_0)$
\begin{equation}
\label{eq:nobigjump}
\lim_{N \to \infty}
\sup_{m:\, \vert m/N -\rho \vert \le \gep}
\bQ_{N,m, \go, \gb} \left(\eta_{1, N} > \gamma N \right) \, =\, 0\, .
\end{equation}
\end{proposition}
\medskip

\begin{proof}
It suffices to prove the result for $\gamma$ small.
Set $A_N= \{\eta_{1,N}> \gamma N\}$. We have
\begin{equation}
\bP_{N,\go, \gb, h} (A_N)\, =\, 
\frac 1{Z_{N,\go, \gb, h} } \sum_{m=1}^N e^{hm} \cZ_{N, m, \go, \gb}\bQ_{N,m, \go, \gb} (A_N)\, ,
\end{equation}
Therefore for every $\rho\in (0,1)$, every $\gep>0$ and every $h$ we have 
\begin{equation}
\label{eq:ggint6-1}
\sup_{m:\, \vert m/N -\rho \vert \le \gep}
\bQ_{N,m, \go, \gb} (A_N) \, \le \, \frac{Z_{N,\go, \gb, h}}{\inf_{m:\, \vert m/N -\rho \vert \le \gep} e^{hm} \cZ_{N, m, \go, \gb}}
\bP_{N,\go, \gb, h} (A_N)\, .
\end{equation}
Choose $h$ such that $\partial_h\tf(\gb, h)=\rho$, so that $\bbP(\dd \go)$-a.s.
\begin{equation}
\lim_{N\to \infty} \frac 1N \log Z_{N,\go, \gb, h}\, =\, \tg(\gb, h)+h\rho\, .
\end{equation}
Therefore we can use  Proposition~\ref{th:Gsurg} to bound  the ratio of partition functions in \eqref{eq:ggint6-1} thus obtaining
that there exists $c_\rho>0$ such that for $N$ larger than a random threshold
\begin{equation}
\sup_{m:\, \vert m/N -\rho \vert \le \gep}
\bQ_{N,m, \go, \gb} (A_N) \, \le \, \exp\left(  c_\rho \gep N\right)
\bP_{N,\go, \gb, h} (A_N)\, .
\end{equation}
By combining this last estimate with \eqref{eq:verylong} we see that for $\gep < \gamma\tf(\gb,h)/(2c_\rho)$ --
and satisfying also the other smallness requirements in Proposition~\ref{th:Gsurg} --
we have that a.s.
\begin{equation}
\lim_{N \to \infty}
\sup_{m:\, \vert m/N -\rho \vert \le \gep}
\bQ_{N,m, \go, \gb} (A_N) \, =\, 0\, .
\end{equation}
\end{proof}

\medskip

\noindent 
\emph{Proof of Theorem~\ref{th:beta>0}(2).}
We proceed like in the $\gb=0$ case, see
\eqref{eq:fordisalso} of Lemma~\ref{th:Zeps}, replacing $\bP(\tau_m=N)$ with  $ \cZ_{N, m, \go, \gb}$
and by using Proposition~\ref{th:Gsurg} instead of Proposition~\ref{th:propG}. We use 
$\rho_h$ the optimizer of the first variational problem 
in \eqref{eq:FHG} and we exploit the hypothesis
 \eqref{eq:defPsi3}.
This way we see that there exists $N_0(\go)\stackrel{\text{a.s}}< \infty$ such that for $N \ge N_0(\go)$ 
\begin{equation}
 \label{eq:fordisalso2}
 \begin{split}
0 \, \le \,  Z^\Psi_{N,\go, \gb, h} - Z^{\Psi,\gep}_{N,\go, \gb, h}\,
&\le \, 
  c \sumtwo{m\in \{1, \ldots,N\}:}{\vert m/N- \rho_h  \vert > \gep }
  \exp\left( (h+b) m +N H(m/N)\right) \bP(\tau_m=N)  \cZ_{N, m, \go, \gb}
  \\
  & \le \, c^2 N  
  \exp\left(2b N + N \sup_{\rho:\, \vert \rho -\rho_h\vert >\gep}\left(  \rho h  + H(\rho)+ \tg(\gb, \rho)\right)\right)
 \, ,
 \end{split}
 \end{equation} 
 where $Z^{\Psi,\gep}_{N,\go, \gb, h}$ is the direct generalization of the analogous quantity in the $\gb=0$ case, see Lemma~\ref{th:Zeps}.
 Since $b$ can be chosen arbitrarily small and by (strict) concavity of $H(\cdot)+ \tg(\gb, \cdot)$
 (note that, since $\gb>0$, $\tg(\gb, \cdot)$ is strictly concave)
  we obtain also in this case that for every $\gep>0$ there exists $p_\gep>0$ such that
\begin{equation}
 \label{eq:fordisalso3}
0 \, \le \,  Z^\Psi_{N,\go, \gb, h} - Z^{\Psi,\gep}_{N,\go, \gb, h}\,
   \le\, \exp\left( N(\tf_H(\gb, h) -p_\gep)\right)\, ,
\end{equation}
for $N$ larger than an a.s . finite random quantity. So the fact that the ratio of 
$Z^{\Psi,\gep}_{N,\go, \gb, h}$ and $Z^\Psi_{N,\go, \gb, h}$ tends a.s. to one takes care of 
\eqref{eq:beta>0contact} because $\rho_h = \partial_h\tf_H(\gb, h)$ for every $h$. 

For \eqref{eq:beta>0} we use the $\gb>0$ version of \eqref{eq:dec-contact}, that is 
\begin{equation}
 \label{eq:dec-contact-beta}
 \bP_{N,\go, \gb, h}^\Psi (\cdot)\, =\, 
 \frac{ \sum_{m=1}^N\exp(hm) \Psi(m, N) \cZ_{N,m, \go, \gb}\bQ_{N, m} (\cdot)}
 {\sum_{m=1}^N\exp(hm) \Psi(m, N) \cZ_{N,m, \go, \gb}} 
 \, .
 \end{equation} 
 Now we fix any $h>h_c^H(\gb)$, so $\rho_h>0$, and we observe that, by \eqref{eq:beta>0contact}, we have that for every $\gep>0$ and a.s.
 \begin{equation}
 \label{eq:dec-contact-b4}
 \bP_{N,\go, \gb, h}^\Psi \left( \eta_{1,N}> \gamma N \right) \stackrel{N \to \infty} \sim 
 \bP_{N,\go, \gb, h}^\Psi \left(\left \vert \sum_{n=1}^N \gd_j - \rho_h N\right \vert \le \gep N , \,   \eta_{1,N}> \gamma N \right)\, .
 \end{equation} 
 We can now insert this event into \eqref{eq:dec-contact-beta} and, by using 
 Proposition~\ref{th:nobigjump}, we readily see that for $\gep$ small the right-hand side of \eqref{eq:dec-contact-b4} vanishes a.s. when $N\to \infty$. 
 Since $\gamma>0$ can be chosen arbitrarily small, we are done. 
 \qed

\section{On free energies and  variational formulas: proof of Theorem~\ref{th:GF}}
\label{sec:var}
\subsection\protect{{\bf On} $\tg (\gb, \rho)${\bf .}}
Recall that the definition \eqref{eq:cZEG} of $ \cZ_{N,m, \go, \gb}$.
For $\rho\in (0, 1]$ we set 
\begin{equation}
\tilde \cZ_{N, \go} (\rho)\,:=\, \min_{m \in\{\lfloor \rho N\rfloor, \lceil \rho N\rceil\}} \cZ_{N,m, \go, \gb}\, ,
\end{equation}
where the set over which the minimum is taken reduces to a singleton if $\rho N$ is integer.
Note that $\cZ_{N,m, \go, \gb}$ is zero if $m=0$, so $\tilde \cZ_{N, \go} (\rho)=0$ whenever $\rho  < 1/N$. 

\medskip

\begin{lemma}
\label{th:superadd}
$(\log \tilde \cZ_{N, \go} (\rho))_{N =1,2, \ldots}$ is super-additive, namely: for every $N_1, N_2 \in \bbN$ we have
\begin{equation}
\label{eq:superadd}
\log \tilde \cZ_{N_1+N_2, \go} (\rho)\, \ge \, \log \tilde \cZ_{N_1, \go} (\rho) +\log \tilde \cZ_{N_2, \theta^{N_1}\go} (\rho)\, ,
\end{equation}
where $(\theta \go)_n=\go_{n+1}$.
\end{lemma}
\medskip

\begin{proof}
First of all remark that for every $b,c \ge 0$ 
\begin{equation}
\lfloor b \rfloor+ \lfloor c \rfloor\, \le \, \lfloor b+c \rfloor \, \le \, \lceil b+c \rceil \, \le \,  \lceil b\rceil+\lceil c \rceil\, ,
\end{equation}
which implies that both the lower and upper integer part of $c+b$ coincide with the sum of a suitably chosen combination 
of upper and/or lower integer parts of $b$ and $c$. For example if neither $b$ nor $c$ is an integer, then  either
 $\lceil b+c \rceil =  \lceil b\rceil+\lceil c \rceil$ or 
 $\lceil b+c \rceil =  \lceil b\rceil+\lfloor c \rfloor=\lfloor b \rfloor+ \lceil c\rceil$. On the other hand, if $b$ is an integer and $c$ is not 
 $\lceil b+c \rceil =    b+\lceil c \rceil=\lfloor b \rfloor+\lceil c \rceil=\lceil b\rceil+\lceil c \rceil=$. The case in which $b$ and $c$ are both integers is of course trivial.  
 
 Then remark also that 
 \begin{equation}
 \log \cZ_{N_1+N_2,m_1+m_2, \go, \gb}\, \ge \,  \log \cZ_{N_1,m_1, \go, \gb} +
 \log \cZ_{N_2,m_2, \theta^{N_1}\go, \gb}\, ,
 \end{equation}
 which follows by restricting  the expectation in the definition of   $\cZ_{N_1+N_2,m_1+m_2, \go, \gb}$
 to the event $\tau_{m_1}=N_1$ and by using the independence of the increments of $\tau$.
 
 Since $\lceil \rho (N_1+N_2)\rceil$ is one among  $\lceil \rho N_1 \rceil+\lceil \rho N_2\rceil$, 
 $\lceil \rho N_1 \rceil+\lfloor \rho N_2\rfloor$  and $\lfloor \rho N_1\rfloor+ \lceil \rho N_2 \rceil$
 and since exactly the same holds true is we switch upper integer parts with lower integer parts,
  we readily see that \eqref{eq:superadd} holds and the proof is complete.
\end{proof}

\medskip

\begin{proposition}
\label{th:Gsurg}
For every $\rho\in [0,1]$ and every $\gb\ge 0$ the limit 
\begin{equation}
\label{eq:Gsurg1}
\limtwo{N \to \infty}{m/N\to \rho}
\frac 1N \bbE \log \cZ_{N,m, \go, \gb} \, =:\, \tg( \gb, \rho)\, ,
\end{equation}
exists and the convergence holds $\bbP(\dd \go)$-a.s, with the same (deterministic) limit,  
without averaging with respect to $\bbP$. Moreover, if we set
\begin{equation}
\label{eq:Gsurg2}
D_{N, \gep, \rho}(\go)\,:=\,
 \suptwo{m\in \{1, \ldots, N\}:}{ \vert m/N- \rho\vert \le \gep} 
\left\vert \frac 1N  \log \cZ_{N,m, \go, \gb} - \tg( \gb, \rho) \right \vert \,  ,
\end{equation}   
for $\rho \in (0,1)$ and $\gep \in (0, \min(\rho/2, (1-\rho)/2)$ we can exhibit a constant $c_{\rho}>0$, with $\sup_{\rho \in[b, 1-b]}<\infty$ for every $b \in (0,1/2)$, such that
$D_{N, \gep, \rho} \le c_{\rho}\gep $ for every $\go$ and $N \ge N_0(\gep, \rho,\go)$, with $\bbP(N_0(\gep, \rho,\go)< \infty)=1$.

For $\rho=0$ we have $\tg(\gb,0)=0$ and there exists $c_0>0$ such that
$D_{N, \gep, \rho}\le c_0 \gep$  for every $\gep \in (0,1/2)$ and for $N$ larger than an a.s. finite  random variable, like above. For $\rho=1$ instead  $\tg(\gb,1)=\log K(1)$ and  there exists $c_1>0$ such that
$D_{N, \gep, 1} \le c_{1}\gep + \tg(1-2\gep) -\tg(1)$ for $N$ larger than a suitable a.s. finite random variable.
\end{proposition}



\medskip

\begin{proof}
For this proof we we fix $\gb\ge 0$ and drop the dependence on $\gb$ from $\cZ_{N,m, \go, \gb}$
We treat first the case $\rho \in (0,1)$. In this case we apply Kingman Sub-additive Ergodic Theorem \cite{cf:King73}, but one has to take care of the fact  
that  $\cZ_{N,m, \go}=0$ for $ \rho N <1$. We deal with this by considering $N_0= N_0(\rho)=\lceil 1/ \rho \rceil$ and by focusing 
on $(\log \tilde \cZ_{nN_0, \go} (\rho))_{n =1,2, \ldots}$. By Lemma~\ref{th:superadd} and by Kingman Sub-additive Ergodic Theorem we have that
\begin{equation}
\label{eq:defGN0}
\lim_{n \to \infty} \frac 1{nN_0} \log \tilde \cZ_{nN_0, \go} (\rho) \,=\, 
\lim_{n \to \infty} \frac 1{nN_0} \bbE \log \tilde \cZ_{nN_0, \go} (\rho) \, =:\, \tg (\gb, \rho)\, ,
\end{equation}
where the first limit is $\bbP(\dd \go)$-a.s..
We now proceed to a \emph{surgical procedure} to compare the partition function $\cZ_{N,m, \go}$
of the systems   that satisfy  $\vert m/N- \rho\vert \le \gep$ with $ \tilde \cZ_{nN_0, \go} (\rho)$, $n$ suitably chosen: we are therefore going to establish 
\eqref{eq:Gsurg2}, from which \eqref{eq:Gsurg1} follows. By the same trick used in the proof of Lemma~\ref{th:superadd} we have that
\begin{equation}
\label{eq:forGsurg}
\log \cZ_{N,m, \go}\, \ge \,  \log \tilde \cZ_{N-\ell_N,m', \go} (\rho)
+  \log \cZ_{\ell_N,m-m', \theta^{N-\ell_N} \go}\, =: \, T_1 +T_2\, ,
\end{equation} 
where $N-\ell_N$ is a multiple of $N_0$
and 
$m':= \lfloor \rho(N -\ell_N) \rfloor$. Recall that  we have $\vert m/N- \rho\vert \le \gep$ and that  we have the constraint that
\begin{equation}
1 \, \le \,  m- m' \, \le \ell_N\, ,
\end{equation}
which simply means that the \emph{second portion of the system} contains at least one contact and no more than its length.  
These requirements are satisfied if 
\begin{equation}
\label{eq:forGsurg2}
\ell_N \, \ge \, \max \left( 
\frac{2+ \gep N}\rho, \frac {1+ \gep N}{1-\rho}
\right)\, ,
\end{equation}
and we can therefore assume that in addition 
\begin{equation}
\label{eq:forGsurg3}
\limsup_N \frac{\ell_N}N  \, \le \, \gep \max \left( 
\frac{1}\rho, \frac {1}{1-\rho}
\right)\,.
\end{equation}
Note also that these definition require $N$ sufficiently large, more precisely $N$ is larger than a multiple of $N_0$ and the proportionality constant 
depends on $\rho$. Everything we claim below is for these values of $N$. 
Let us remark now that for the term $T_1$ in 
\eqref{eq:forGsurg} we have 
\begin{equation}
\begin{split}
\frac 1N \log \tilde \cZ_{N-\ell_N, \go} (\rho) \, &\ge \,
\left(1-\frac{\ell_N} N \right) \frac 1{N - \ell_N} \log  \tilde \cZ_{N-\ell_N, \go} (\rho)
\\& \ge \, 
\left(1-\frac{\ell_N} N \right)
 \tg(\gb, \rho) - \left \vert 
 \frac 1{N-\ell_N} \log \tilde \cZ_{N-\ell_N, \go} (\rho) - \tg( \gb, \rho)
\right\vert  \\
&\ge \,
\left(1-c \gep \right)
 \tg(\gb, \rho)
- \left \vert 
 \frac 1{N-\ell_N} \log\tilde \cZ_{N-\ell_N, \go} (\rho) - \tg( \gb, \rho)
\right\vert 
 \, ,
 \end{split}
\end{equation}
with $c=2\max(1/\rho,1/(1-\rho))$, cf. \eqref{eq:forGsurg3} (this requires $N$ sufficiently large),
and the rightmost term that vanishes a.s. as $N \to \infty$. 
 
We have now to bound $T_2$ from below. We proceed via a rough bound that consists in selecting only one renewal trajectory (simply the one for which the 
$m-m'$ contacts are at the right end of the system):
\begin{multline}
\label{eq:roughboundbG}
\log \cZ_{\ell_N,m-m', \theta^{N-\ell_N} \go}\, \ge \, 
\log\cZ_{\ell_N-(m-m'-1),1, \theta^{N-\ell_N} \go} + \sum_{j=1}^{m-m'-1} \log \cZ_{1,1, \theta^{N-j} \go} 
\\
\ge \, \log K (\ell_N-(m-m'-1))+ (m-m'-1)\log K(1) + \gb  \sum_{j=0}^{m-m'-1}\go_{N-j} 
\\ \ge -c \gep N + \gb  \sum_{j=0}^{m-m'-1}\go_{N-j} \, \ge \, 
-c \gep N - \gb  \sum_{j=0}^{\ell_N-1}\vert \go_{N-j}\vert 
\, ,
\end{multline}  
where in  the step before the last $c>0$ depends on $\rho$ and we have simply used that 
$m-m' \le \ell_N= O(\gep N)$, see \eqref{eq:forGsurg2} and \eqref{eq:forGsurg3}.
Note that $ \bbE \sum_{j=0}^{\ell_N-1}\vert \go_{N-j}\vert =\ell_N \bbE[\vert \go_1\vert] = O(\gep N)$ and 
that, by an elementary Large Deviation bound via exponential Markov inequality, we see that
$\bbP(  \sum_{j=0}^{\ell_N-1}\vert \go_{N-j}\vert \ge 2 \ell_N \bbE[\vert \go_1\vert] )
\le \exp(-c \ell_N)\le \exp(-c' \gep N)$ so that, by Borel-Cantelli, 
$\sum_{j=0}^{\ell_N-1}\vert \go_{N-j}\vert \le 2 \ell_N=O( \gep N)$ for $N$ sufficiently large, how large may depend on $\go$.

The upper bound is obtained exploiting the same idea: the first step is to observe that 
\begin{equation}
\log \cZ_{N,m, \go}\, \le \, \log  \tilde \cZ_{N+\ell_N, \go} (\rho) - \min_{m'\in \{m_+, m_-\}}
\log \cZ_{\ell_N, m'-m, \theta^N\go}\, ,
\end{equation} 
where $m_-= \lfloor \rho (N+\ell_N)\rfloor$ and $m_+= \lceil \rho (N+\ell_N)\rceil$.
Once again, the first term on the right-hand side is controlled using \eqref{eq:defGN0} and we need a lower bound on 
$\log \cZ_{\ell_N, m'-m, \theta^N\go}$, like before. With the same procedure we obtain 
\begin{equation}
 \min_{m'\in \{m_+, m_-\}}
\log \cZ_{\ell_N, m'-m, \theta^N\go} \, \ge \, -c \gep N - \gb  \sum_{j=0}^{\ell_N-1}\vert \go_{N-j}\vert \,.
\end{equation} 
This term can be bounded a.s.  precisely like for the  lower bound, and, by putting upper and lower bound together we obtain the bound for $\rho \in (0,1)$ on 
$D_{N, \gep, \rho}(\go)$, cf. \eqref{eq:Gsurg2}, claimed in Proposition~\ref{th:Gsurg}.
Note that this bound directly implies \eqref{eq:Gsurg1}.

\smallskip

For the case $\rho=0$ we can use the same trick as in \eqref{eq:roughboundbG} to get the lower bound 
\begin{equation}
\log \cZ_{N,m, \go}\, \ge \, \log K(N-m+1) + (m-1) \log K(1) - \gb \sum_{j=N-m+1}^N \vert \go_j\vert\, ,
\end{equation} 
and it is straightforward to see that $-\log \cZ_{N,m, \go}$ is bounded above by a constant time $\gep$
plus a random contribution that is also $O(\gep N)$ both in $\bbL^1$ and a.s..
So $(\log \cZ_{N,m, \go})_-$ is under control and it suffices to remark that
$\bbE \cZ_{N,m, \go} = \exp(\gl(\gb)m)\bbP (\tau_m=N)\le   \exp(\gl(\gb)m)$ which is in turn bounded by   $\exp(\gl(\gb)\gep N)$.
Therefore $\bbE \log \cZ_{N,m, \go} \le 2\gl(\gb)N \gep$ and Borel-Cantelli yields also that 
$ \log \cZ_{N,m, \go} \le 2\gl(\gb) m\le \gl(\gb)N \gep$ for $N$ larger than a constant that may depend on $\go$. 
This completes the proof for $\rho=0$. 

For the case $\rho=1$ we write
\begin{equation}
\label{eq:pre-fte7}
\log \cZ_{N,m, \go}\, =\, \log \bP \left( \tau_m =N\right)+\log
 \bE_{N, m} \left[ \exp\left( -\gb  \sum_{j=1}^N  \go_j (1-\gd_j) \right)\right]
 + \gb \sum_{j=1}^N  \go_j  \,,
\end{equation}
where  $\bQ_{N,m} (\cdot)= \bP(\cdot \vert  \tau_m=N)$, like in Section~\ref{sec:path-beta0}. By Proposition~\ref{th:propG}, notably \eqref{eq:lemG-0}, 
we have
 $ \vert  \log \bP \left( \tau_m =N\right)-\tg(1)\vert \le 
\tg(1-2\gep)$ for $N$ sufficiently large. The last term is also easily disposed of since by standard estimates for IID sequence of centered variables in $\bbL^p$ for every $p$ we have that this term is a.s. $O(N^{c})$, any $c>1/2$.

We are therefore left with controlling the second term in \eqref{eq:pre-fte7}.
By Jensen inequality have the lower bound
\begin{equation}
\label{eq:skdfte7}
\log
 \bE_{N, m} \left[ \exp\left(-\gb  \sum_{j=1}^N  \go_j (1-\gd_j) \right)\right]\, \ge 
 -\gb  \sum_{j=1}^N  \go_j u_{N,m}(j)
 \, ,
\end{equation}
with $ u_{N,m}(j)=(1- \bE_{N, m}  [\gd_j])$. The bound in $\bbL^1$ is obvious
because $\sum _j  u_{N,m}(j) = N-m \le \gep N$. For the a.s. bound  the
 Markov inequality yields 
\begin{multline}
\bbP \left(  \sum_{j=1}^N \go_j u_{N,m}(j) \, \ge \,  \gep N \right)
\, \le \, \exp\left(- t \gep N + \sum_{j} \gl(u_{N,m}(j)t) \right)
\\
 \le \, \exp\left(- t \gep N + \sum_{j} (u_{N,m}(j))^2t^2 \right)
 \, \le \, \exp\left(- t \gep N +\gep N t^2\right) \, \le \, \exp\left(- c \gep N \right)\, ,
\end{multline}
where we have used that $\gl(u) \sim u^2/2$ for $u$ small, so $\gl(u)\le u^2$ for $\vert u \vert \le u_0$, and $u_{N,m}(j)\le 1$
as well as $\sum _j  u_{N,m}(j) = N-m \le \gep N$. In the last step $c$ is the maximum of $t-t^2$ for $t \in [0, u_0]$.
The Borel-Cantelli Lemma warrants that the quantity in \eqref{eq:skdfte7} is bounded below by $-2\gb \gep N$ for $N$
larger than a random threshold.

For the upper bound it suffices to remark that
\begin{equation}
\label{eq:skdfte7-1}
\bbE \bE_{N, m} \left[ \exp\left( \gb  \sum_{j=1}^N  \go_j (1-\gd_j) \right)\right]\, =\, \exp( (N-m) \gl(\gb))\, 
\le \, \exp(\gl(\gb) \gep N )\, ,
\end{equation}
which, by using again the Markov inequality and Borel-Cantelli, yields the a.s. bound we are looking for.
\end{proof}

\medskip

\begin{proposition}
\label{th:Gconcave}
$\tg(\gb, \cdot)$ is concave on $[0,1]$ and it is continuous up to the boundary. Moreover 
\begin{equation}
\label{eq:Gboundary1}
\lim_{\rho \nearrow 1} \frac{\tg(\gb,1)-\tg(\gb, \rho)}{1-\rho}\, =\, -\infty\,.
\end{equation}
\end{proposition}

\medskip

\begin{proof}
Choose $\gl \in (0,1)$ and  $\rho_1, \rho_2 \in [0,1]$ with $\rho_1< \rho_2$. We have
\begin{multline}
\frac 1N  \bbE \log \cZ_{N, \lceil \gl \rho_1+(1-\gl) \rho_2 N \rceil, \go } \, \ge 
\\
\frac 1N  \bbE \log \cZ_{\lceil \gl N \rceil, \lceil \gl \rho_1 N \rceil, \go }+
\frac 1N  \bbE \log \cZ_{N-\lceil \gl N \rceil, \lceil \gl \rho_1+(1-\gl) \rho_2 N \rceil- \lceil \gl \rho_1 N \rceil, \go }\,.
\end{multline}
By Proposition~\ref{th:Gsurg} we can pass to the limit $N \to \infty$ and we obtain 
\begin{equation}
\tg\left(\gb, \gl \rho_1 +(1-\gl) \rho_2\right) \, \ge \, \gl \tg(\gb, \rho_1) +(1-\gl) \tg(\gb, \rho_2)\, ,
\end{equation}
so $\tg(\cdot)$ is concave, hence continuous because it is bounded. Both the continuity at $0$ and $1$,
with $\tg(\gb, 0)=0$ and $\tg(\gb,1)=\log K(1)$,  follow directly from the estimates  in Proposition~\ref{th:Gsurg} 
(we observe that   the continuity in $(0,1)$ can be extracted directly from Proposition~\ref{th:Gsurg} as well).

For what concerns \eqref{eq:Gboundary1} we need an adequate lower bound on $\tg(\gb, \rho)$. 
This follows by taking the expectation of both sides of 
\eqref{eq:pre-fte7} and \eqref{eq:skdfte7}. This way we obtain 
\begin{equation}
\label{eq:newfte7}
\bbE \log \cZ_{N,m, \go}\, \ge \, \log \bP \left( \tau_m =N\right) \,,
\end{equation}
and the general bound $\tg(\gb, \rho) \ge \tg(0, \rho)$. This inequality becomes an equality at $\rho=1$ (this follows once again from \eqref{eq:pre-fte7}) and therefore 
$\tg(\gb,1)- \tg(\gb, \rho) \le \tg(0,1)- \tg(0, \rho)$, so that the claim follows from the analogous claim for the case $\gb=0$.
\end{proof}

\subsection{Proof Theorem~\ref{th:GF}}

\begin{proof}
Let us make the preliminary remark that it suffices to show a.s. convergence because 
$( \log Z_{N, \go, \gb, h}^\Psi /N)_{N=1,2, \ldots}$ is uniformly integrable and therefore we have also convergence in $\bbL^1$. 
Uniform integrability can be established 
by making  upper and lower bounds on $Z_{N, \go, \gb, h}^\Psi $ in the spirit of the repeated estimates we used in  the proof of Proposition~\ref{th:Gsurg} (but what suffices here is substantially rougher), so   
one readily  sees that there exists $C>0$ (that depends on $\Psi$, on $K(\cdot)$ and on $h$) such that
\begin{equation}
\frac 1N \left \vert  \log Z_{N, \go, \gb, h}^\Psi  \right\vert \, \le \, \frac \gb N \sum_{j=1}^N \vert \go_j \vert  + C\, .
\end{equation}
Since the expectation of the square  of the right-hand side is bounded uniformly in $N$,  uniform integrability is proven. 

We now proceed with proving a.s. convergence by suitable lower and upper bounds on the sequence. Remark that the expected limit 
$\sup_{\rho \in [0,1]} \left( h \rho +H(\rho) + \tg(\gb, \rho)\right)$
 is in fact a maximum which is uniquely achieved at $\rho_h \in [0, 1)$ by the assumptions on $H(\cdot)$ and by what we have proven on $\tg(\gb, \cdot)$ (that is, Proposition~\ref{th:propG} and Proposition~\ref{th:Gconcave}).

For the lower bound it suffices to remark that for every $\rho$ \eqref{eq:basic-decomp}  (with the notation \eqref{eq:cZEG}) yields
thanks to \eqref{eq:defPsi1} of Definition~\ref{def:Psi}
\begin{equation}
 \label{eq:basic-decomp.1}
 Z_{N, \go, \gb, h}^\Psi \,\ge \, 
   c(b) \exp\left(mh+NH(m/N)-bN\right)  
\cZ_{N,m, \go, \gb}\,,
\end{equation}  
With $b$ that can be chosen arbitrarily small. 
If  $\rho_h\in (0,1)$ it suffices to choose $\rho=\rho_h$ and pass to the limit: by Proposition~\ref{th:Gsurg} we obtain that a.s.
\begin{equation}
 \label{eq:frbasic-decomp.1}
 \liminftwo{N\to \infty:}{m/N\sim \rho_h} \frac 1N Z_{N, \go, \gb, h}^\Psi \,\ge \, 
   \rho_h h+H(\rho_h)+ \tg( \gb,\rho_h) -b\,.
\end{equation}  
Since $b>0$ is arbitrary, we are done for the case $\rho_h\in (0,1)$. If 
$\rho_h=0$ we can repeat the same argument for $\rho=\rho_j$, with $\rho_j \searrow 0$, and the lower bound analysis is complete.

For the upper bound we use \eqref{eq:defPsi3} of Definition~\ref{def:Psi}
\begin{equation}
 \label{eq:basic-decomp.2}
 Z_{N, \go, \gb, h}^\Psi \,\le \, c(b)\sum_{m=1}^N
   \exp\left(mh+NH(m/N)+ bN \right)  
\cZ_{N,m, \go, \gb}\,,
\end{equation}  
so
\begin{equation}
\limsup_N \frac 1N \log Z_{N, \go, \gb, h}^\Psi \,\le \, \limsup_N\max_{m=1, \ldots, N}  \left(
\frac m N h +H\left( \frac m N \right) + \frac 1N \log \cZ_{N,m, \go, \gb} +b
\right)  \, .
\end{equation}
To deal with the maximum we fix 
a small value positive value of $\tilde \rho$ and 
the grid of densities $\rho_j :=\tilde \rho+  j(1-2 \tilde \rho)/M$, $M$ a positive integer, for $j=0, 1, \ldots, M$: we can therefore group the maximum into 
$M+3$ blocks.
We can now apply   Proposition~\ref{th:Gsurg} in taking the limit $N \to \infty$, with $\gep=1/M$ for for the blocks corresponding to the densities 
$\rho_0= \tilde \rho, \rho_1, \ldots , \rho_M= 1- \tilde \rho$, and with $\gep= \tilde \rho$ fro the two boundary blocks.  
It is then a matter of sending first $M$ to $\infty$ and $ \tilde \rho $ to zero. 
Since also $b$ can be chosen arbirarily small,  we conclude that a.s.
\begin{equation}
\limsup_N \frac 1N \log Z_{N, \go, \gb, h}^\Psi \,\le \, \sup_{\rho \in [0,1]} \left( h \rho +H(\rho) + \tg(\gb, \rho)\right)  \, .
\end{equation}
This completes the proof of Theorem~\ref{th:GF}.
\end{proof}

\appendix

\section{Circular DNA models}
\label{sec:circDNA}

Three  models are considered  in \cite{cf:bar1,cf:bar2}:
two of them enter the  setting of our work, the third demands a 
$\Psi(m, N)$ that vanishes when $m/N<1/2$. 
A wider set-up that encompasses all models in in \cite{cf:bar1,cf:bar2} can be found in \cite[Ch.~4]{cf:thesis-BH} (to which he refer also for a more complete presentation on circular DNA and renewals).
\smallskip

The two models are:
\smallskip

\begin{enumerate}
\item \emph{Model with overtwist}: 
\begin{equation}
\Psi(m, N)\, =\, \exp\left(-\chi (N-m)^2/m\right)\, ,
\end{equation}
with $\chi>0$. In this case $Q(m, N)\equiv 1$ and $H(\rho)= \chi (1-\rho)^2/\rho$. 
\item \emph{Model with overtwist and supercoils}:
\begin{equation}
\Psi(m, N)\, =\, 
\sum_{n=0}^m 
\frac 1{2^m} \binom{m}{n}
\exp\left( n w - \chi\frac{(N-m-n)^2}m\right) 
\, ,
\end{equation}
with $\chi$ and $w$ positive constants. In this case 
\begin{equation}
H(\rho) =  \sup_{\zeta \in[0, \rho] } \psi (\zeta, \rho) = \psi (\zeta_0(\rho), \rho),
\end{equation}
where
\begin{equation}
\psi (\zeta, \rho) := \zeta w -\rho \log 2 - \chi \frac{(1 - \rho -\zeta)^2}{\rho} +  \rho \log \rho   -  \zeta\log( \zeta)     - (\rho -\zeta) \log(\rho -\zeta).
\end{equation}
$Q(m, N)$ is therefore (implicitly) defined and  one can derive the asymptotic behavior  for $N\to \infty$ and $m/N$ asymptotically constant:
\begin{equation}
Q(m,N) \sim q\left(\frac{m}{N}\right),\  \text{ with }\   q(\rho) := \sqrt{\frac{\rho}{ \zeta_0( \rho)(\rho- \zeta_0( \rho))|\partial_\zeta^2 \psi (\zeta_0( \rho),\rho)|} }.
\end{equation}
The function $\psi$ is concave on the convex domain $\{0 \le  \zeta \le  \rho \le  1\}$. Thus, $H$ is also concave. Moreover $H$ is analytic on $(0,1)$. 
\end{enumerate}

\smallskip

We remark that in both examples $H(0) = -\infty$ and (of course) $H'(0)= \infty$. So the circular DNA models are localized for all values of the parameters
and they display a big-jump transition if $\gb=0$ and $\ga >1$. The third model in \cite{cf:bar1,cf:bar2}, that we do not consider, is the $\chi=\infty$ limit of the second model.

\section{On the strict convexity of the disordered pinning free energy}
\label{sec:unifconvex}

\begin{theorem}
\label{th:unifconvex}
Consider the $\Psi\equiv 1$ model, that is the disordered pinning model. 
For every $\gb\ge 0$ and every $h >h_c(\gb)$  we have $\partial^2_h \tf (\gb, h) >0$. 
\end{theorem}
\medskip

\begin{proof} Let us remark that for $\gb=0$ the result can be established by explicit computations, but the proof that we give here works for the $\gb=0$ case as well.
In this proof $\bP_{N, \go}= \bP_{N, \go, \gb, h}$ and $ \var_{N, \go}$ is the variance with respect to  $\bP_{N, \go}$.
We know from \cite[Proof of Theorem~2.1]{cf:GTalea} that for $h >h_c(\gb)$
\begin{equation}
\label{eq:here5.1}
\partial^2_h \tf (\gb, h) \, =\, \lim_{N \to \infty} \frac 1N \bbE\,  \var_{N, \go} \left(
\sum_{j=1}^{N-1} \gd_j
\right)\,.
\end{equation}
We are going to condition on even sites, so let us replace $N$ with $2N$ and let us denote by   $\cF_e$ the $\gs$-algebra generated by 
$\gd_j$ with $j$ even: $\var_{N, \go, e} (\cdot)$ is going to denote the variance with respect to $\bP_{N, \go}(\, \cdot \, \vert \cF_e)$.
By Jensen's inequality
\begin{equation}
\label{eq:here5.2}
 \var_{2N, \go} \left(
\sum_{j=1}^{2N-1} \gd_j
\right) \, \ge \, \bE_{2N, \go} \left[
 \var_{2N, \go, e} \left(
\sum_{j=1}^{N} \gd_{2j-1}
\right)
\right]\, .
\end{equation}
We know consider the conditional variance on the set $E_\gs:= \{ \tau : \, \gd_{2j} =\gs_j $ for $j =1,2, \ldots, N-1\}$
for every $\gs \in \{0,1\}^{\{1,\ldots, N-1\}}$. We set $n(\gs):= \sum_{j=1}^{N-1} \gs_j$ and $\ell_0:=0$ and we define
iteratively  $\ell_{j+1}= \min\{ \ell > \ell_j: \, \gs _\ell=1\}$ for $j \le n(\gs)-1$. We then redifine $\ell_j$ to be $2\ell_j$ and  set also $\ell_{n(\gs)+1}:=2N$. Therefore
$\ell_0, \ell_1, \ldots,  \ell_{n(\gs)+1}$ are the $n(\gs)$ pinned even sites, plus $0$ and $2N$ that are pinned from the start. 
Note that
\begin{equation}
\sum_{j=1}^N \gd_{2j-1} \, =\, \sum_{k=1}^{n(\gs)+1} \sum_{j=1+\ell_{k-1}/2}^{\ell_k/2} \gd_{2j-1}\, ,
\end{equation}
and remark that, under $\bP_{N, \go}(\, \cdot \, \vert \cF_e)(\tau)$ with $\tau \in E_\gs$, 
the random variables 
\begin{equation}
\left(  \sum_{j=1+\ell_{k-1}/2}^{\ell_k/2} \gd_{2j-1}
\right)_{k=1, \ldots, n(\gs)+1}\, ,
\end{equation}
are independent. Therefore on $E_\gs$
\begin{multline}
\var_{2N, \go, e} \left(
\sum_{j=1}^{N} \gd_{2j-1}
\right) \, =\, 
\\
\sum_{k=1}^{n(\gs)+1} 
\var_{2N, \go, e} \left( \sum_{j=1+\ell_{k-1}/2}^{\ell_k/2} \gd_{2j-1}
\right) 
\, \ge\, 
\sumtwo{k=1, \ldots, n(\gs)+1}{\ell_k -\ell_{k-1}=2}
\var_{2N, \go, e} \left(\gd_{\ell_k-1} \right)\,. 
\end{multline}
Since $\gd_{\ell_k-1}$, under the conditional measure we are considering, is just a Bernoulli random variable 
with parameter (we use the short-cut notation $\go=\go_{\ell_k-1}$)
\begin{equation}
p( \go)\, :=\, \frac{K(1)^2 \exp(h + \gb \go  )}{K(1)^2 \exp(h + \gb\go  ) +K(2)}\, , 
\end{equation}
we see that for $k$ such that $\ell_k-\ell_{k-1}=2$
\begin{equation} 
\var_{2N, \go, e} \left(\gd_{\ell_k-1} \right)\, =\, p\left( \go_{\ell_k-1} \right)
\left( 1- p\left( \go_{\ell_k-1} \right)\right)\, =: \, \gs ^2\left( \go_{\ell_k-1} \right) ,
\end{equation}
and therefore  
\begin{equation}
\var_{2N, \go, e} \left(
\sum_{j=1}^{N} \gd_{2j-1}
\right) \, \ge \, \bE_{2N, \go} \left[ \sum_{k=0}^{N-1} \gd_{2k} \gd_{2k+2}\gs^2\left( \go_{\ell_k-1} \right) \right]\, .
\end{equation}
Now we set $\gs_\star^2(L):= \inf \{ \gs^2(\go):\, \vert \go\vert \le L\}>0$. We remark that $\gs_\star^2(L)>0$   for every  $L>0$, but in what follows  we are forced to work with $L$ such that  $\bbP (\vert \go \vert < L)>0$, that is for  $L$ above a threshold.
With this notation
\begin{equation}
\begin{split}
\bbE \var_{2N, \go, e} \left(
\sum_{j=1}^{N} \gd_{2j-1}
\right) \, &\ge \, \gs_\star^2(L) \bbE \bE_{2N, \go} \left[ \sum_{k=0}^{N-1} \gd_{2k}\gd_{2k+2} \ind_{\vert \go_{2k+1}\vert \le L} \right]
\\
& \ge \,  \gs_\star^2(L)
\left(  \bbE \bE_{2N, \go} \left[ \sum_{k=0}^{N-1} \gd_{2k}\gd_{2k+2} \right]
- N \bbP\left( \vert \go \vert >L\right)\right) 
\, ,
\end{split}
\end{equation}
and we are left with showing that  
\begin{equation}
\label{eq:leftwith}
q\, :=\, \liminf_N \frac 1N \bbE \bE_{2N, \go} \left[ \sum_{k=0}^{N-1} \gd_{2k}\gd_{2k+2} \right] \, >0\, ,
\end{equation} 
because it suffices to choose $L$ so that $\bbP\left( \vert \go \vert >L\right) \le q/2$
to obtain, see  \eqref{eq:here5.1}-\eqref{eq:here5.2},  that $\partial_h^2 \tf (\gb, h)\ge \gs_\star^2(L) q/4>0$.

In order to  establish \eqref{eq:leftwith} we want to show that the quantity under analysis is bounded below by 
$\lim_NN^{-1} \bbE \bE_{2N, \go}[\sum_{j=1}^{2N} \gd_j]$, which is equal to  $2 \partial \tf (\gb, h)>0$, times a positive constant.
This can be done by explicit estimates, but for sake of conciseness we use that, for any choice of a sequence $(b_N)_{N \in \bbN}$
of positive integer numbers satisfying  $\lim_N b_N= \infty$ and $\lim_N{b_N}/N=0$, by \cite[Theorem~2.2]{cf:GTalea} we have 
uniformly on $k\in [b_N, N-b_n]\cap \bbN$
\begin{equation}
\label{eq:uniformA}
\lim_{N \to \infty} \bbE   \bE_{2N, \go}[\gd_k] \, =\, \partial_h\tf (\gb, h)
\ \ \ \ \text{ and } \ \ \ 
\lim_{N \to \infty} \bbE  \bE_{2N, \go}[\gd_k\gd_{k+2}]\, =:\, \tl(\gb,h)\, ,
\end{equation}
where the second statement is just the existence of the limit and \eqref{eq:leftwith} follows once $ \tl(\gb,h)>0$ is shown.
For this we write $\bE_{2N, \go}[\gd_k\gd_{k+2}]=  \bE_{2N, \go}[\gd_k]   \bE_{2N, \go}[\gd_{k+2}\vert \gd_{k}=1 ]$
and 
\begin{equation}
 \bE_{2N, \go}[\gd_{k+2}\vert \gd_{k}=1 ]
\, =\,  \frac{Z_{2, \theta^{k}\go, \gb ,h}  Z_{N-k-2, \theta^{k+2}\go, \gb ,h} }
{Z_{N-k, \theta^{k}\go, \gb ,h} }\, \ge \, C \exp\left(- \gb \left(\vert \go_{k+1}\vert + \vert \go_{k+2}\vert \right) \right)\, ,
\end{equation}
where the constant $C>0$ depends on $h$ and on $K(\cdot)$: this estimate is a standard surgery procedure (\cite[Ch.~2]{cf:GB}, full details can be found 
in  \cite[Sec.~5.5]{cf:thesis-BH}) for which one uses notably the regularly varying character of $K(\cdot)$. 
Therefore 
\begin{equation}
  \bE_{2N, \go}[\gd_k\gd_{k+2}]\, \ge Ce^{-2 \gb L}  \bE_{2N, \go}[\gd_k] \left(1-\ind_{\vert \go_{k+1}\vert + \vert \go_{k+2}\vert > L }
  \right)\, ,
\end{equation}
and, in turn, we have
\begin{equation}
  \bbE\bE_{2N, \go}[\gd_k\gd_{k+2}]\, \ge\,  Ce^{-2 \gb L}  \left(\bbE\bE_{2N, \go}[\gd_k]- \bbP
  \left( \vert \go_{1}\vert > L \right)\right)\,  .
  \end{equation}
  It suffices now to choose $L$ so that  $\bbP
  \left( \vert \go_{1}\vert > L \right) \le \partial_h\tf(\gb, h)/2$ to obtain that, uniformly in $k$ like in \eqref{eq:uniformA}, we have
 \begin{equation}
 \liminf_N
  \bbE\bE_{2N, \go}\left[\gd_k\gd_{k+2}\right]\, \ge\, \frac 12 Ce^{-2 \gb L}  \partial_h\tf(\gb, h) \, >\, 0\, ,
  \end{equation}
  and we are done.
\end{proof}

\medskip

We include here 
 the result proved under restrictive conditions in Remark~\ref{rem:de-loc}.

\medskip

\begin{proposition}
\label{th:de-loc2}
For every $\gb \ge 0$ and every $h$ 
\begin{equation}
\tf_H(\gb, h)\,  \ge\,  H(0)\, .
\end{equation}
\end{proposition}

\medskip

\begin{proof}
We can  assume $H(0)> -\infty$ and, 
with $b>0$ and $u$, $v$ and $c(b)$ like in \eqref{eq:defPsi1} of Definition~\ref{def:Psi}, we obtain that 
\begin{multline}
 Z^\Psi_{N, \go \gb ,h} \, \ge \, c(b) \exp\left(N\min_ {\rho \in (0, b]} H(\rho) - bN\right)
 \\ \bE\left[
 \exp\left(\gb \sum_{j=1}^N (\gb \go_j +h) \gd_j\right); \, \tau_{\lfloor b N \rfloor} =N, \, 
 \frac{\tau} N\cap\left( (0, b)\cup (1-b, 1)\right) = \emptyset\right]
 \\ 
 \, , 
 \end{multline}
 a bound that is obtained simply by restricting the partition function to the renewals with $\lfloor b N \rfloor$ contacts 
 and all at distance at least $bN$ from the boundary.
 Therefore 
 \begin{equation}
 \liminf_{N \to \infty}  \frac 1N \bbE \log Z^\Psi_{N, \go \gb ,h} \, \ge \, H(0) + \liminf_{N \to \infty}  \frac 1N \bbE \log 
 \bE\left[
 \exp\left( \sum_{j=1}^N (\gb \go_j +h) \gd_j\right); \,    E_{N, b} \right] \, ,
 \end{equation} 
 with $E_{N, b}:= \{ \tau_{\lfloor b N \rfloor} =N, \, 
 (\tau/N)\cap\left( (0, b)\cup (1-b, 1)\right) = \emptyset\}$. Using $\bP'(\cdot):= \bP(\cdot \vert E_{N, b})$ we see that by Jensen's inequality the quantity of which we take 
  inferior limit in the right-hand side  of the last expression is bounded below by 
  \begin{equation}
     \frac 1N  
 \bbE \left[ \sum_{j=1}^N (\gb \go_j +h) \bE'[\gd_j] \right] 
 +    \frac 1N   \log \bP\left(   E_{N, b}\right)
 \, =\,  h \frac{ \lfloor bN\rfloor} N+   \frac 1N   \log \bP\left(   E_{N, b}\right).
  \end{equation} 
  Of course the limit of the first term is $hb$, which can be made arbitrarily small by choosing $b$ small. The remaining term 
   is bounded below, for $N \to \infty$,  by a (negative) quantity that vanishes as 
  $b \searrow 0$ because 
 $ \bP(E_{N, b})$ is bounded below by  $K(\lceil bN\rceil)^2$ times 
 $\bP ( \tau_{\lfloor b N \rfloor-2}= N- 2\lceil bN\rceil)$, so, by Proposition~\ref{th:propG},  $\liminf_N (1/N) \log  \bP(E_{N, b})=0$ for $\ga>1$ and it vanishes as $b \searrow 0$ for $\ga \in (0,1]$.
  \end{proof}

 \section*{Acknowledgements}
We are grateful to Quentin Berger and Hubert Lacoin for several exchanges. We thank in particular  Hubert Lacoin
for  suggesting the argument of proof of Theorem~\ref{th:unifconvex}.
G.G. also acknowledges support from grant ANR-15-CE40-0020.

\end{document}